\newif\ifLONG
\LONGtrue

\documentclass{article}
\usepackage{tikz}
\usepackage{subcaption}
\usepackage{amsfonts}
\usepackage{amssymb}
\usepackage{graphicx}
\usepackage{amsmath}
\usepackage{amsthm}
\usepackage{pifont}
\usepackage{xcolor}
\usepackage{enumitem}
\usepackage{mathtools}
\usepackage[normalem]{ulem}

\usepackage{tikz}
\usepackage{pgfplots}
\pgfplotsset{compat=1.18}

\usepackage[colorlinks=true,linkcolor=blue,citecolor= red,bookmarksopen=true,urlcolor=blue]{hyperref}
\nonstopmode
\newtheorem{theorem}{Theorem}[section]
\newtheorem{proposition}[theorem]{Proposition}
\newtheorem{lemma}[theorem]{Lemma}
\newtheorem*{lemma*}{Lemma}

\newtheorem{corollary}[theorem]{Corollary}

\newtheorem{definition}[theorem]{Definition}
\newtheorem{notation}[theorem]{Notation}

\newtheorem{assumption}[theorem]{Assumption}

\theoremstyle{remark}

\newtheorem{remark}[theorem]{Remark}
\theoremstyle{plain} 
\newtheorem*{standass}{Standing Assumptions}
\newtheorem*{standingassumptionII*}{Standing Assumption II}
\newtheorem*{stupa}{Setup A}
\newtheorem*{stupb}{Setup B}

\newcommand{\probp}{{P}}
\newcommand{\probq}{{Q}}
\newcommand{\R}{{\mathbb R}}

\newcommand{\norm}[1]{\left\lVert#1\right\rVert}
\newcommand\abs[1]{\left|#1\right|}

\newcommand\Epo[1]{{E}_{P^i} \left[#1\right]}

\newcommand{\rn}[2]{\frac{\mathrm{d}#1}{\mathrm{d}#2}}

\newcommand{\fit}{{\mathcal{F}^{i}_t}}

\newcommand{\fib}{{\mathbb{F}^{i}}}

\newcommand{\NCA}{\mathbf{NCA}(\mathcal{Y})} 
\newcommand{\NCFL}{\mathbf{NCFL}(\mathcal{Y})}

\renewcommand{\emph}[1]{\textit{#1}}
\newcommand{\red}[1]{\textcolor{red}{#1}}

\newcommand{\blue}[1]{{\color{blue}#1}}

\newcommand{\mcY}{\mathcal {Y}}
\newcommand{\mcF}{\mathcal {F}}

\newcommand{\Me}{\mathcal M_e(\mcY)}

\newcommand{\Pp}{\mathbb{P}}

\usepackage{graphicx} 

\title{When cooperation is beneficial to all agents}
\author{
Alessandro Doldi\thanks{Dipartimento di Matematica, Università degli Studi di Milano, Via Saldini 50, 20133 Milano, Italy,
\emph{alessandro.doldi@unimi.it}.}, $\quad$ Marco Frittelli\thanks{Dipartimento di Matematica, Università degli Studi di Milano, Via Saldini 50, 20133 Milano, Italy,
\emph{marco.frittelli@unimi.it}.  }, $\quad$
Marco Maggis\thanks{Dipartimento di Matematica, Università degli Studi di Milano, Via Saldini 50, 20133 Milano, Italy,
\emph{marco.maggis@unimi.it}.}
}

\date{\today}

\begin{document}
\maketitle
\begin{abstract}

    \noindent Within a general semimartingale framework, we study the relationship between collective market efficiency and individual rationality. We derive a necessary and sufficient condition for the existence of (possibly zero-sum) exchanges among agents that strictly increase their indirect utilities and characterize this condition in terms of the compatibility between agents’ preferences and collective pricing measures. The framework applies to both continuous- and discrete-time models and clarifies when cooperation leads to a strict improvement in each participating agent’s indirect utility.

\end{abstract}

\section{Introduction}
\label{iNTRO}

Classical asset–pricing theory is largely framed around the actions of a single agent operating in a frictionless market. Since the seminal works of Kreps \cite{Kreps81}, Harrison--Kreps \cite{HarrisonKreps79} and Harrison--Pliska \cite{HarrisonPliska81}, market efficiency and pricing have been characterized by individual no–arbitrage and martingale measures. Modern financial environments increasingly exhibit features where cooperation, risk-sharing, segmentation (by assets, frictions, or information) and collective behavior redefine the boundaries of what is considered an \textit{arbitrage} or a \textit{replicable claim}. This paper advances the theory of \textit{Collective Finance}, as developed in \cite{BDFFM25}, \cite{DFM25} and \cite{F25}. In this framework, a market might appear arbitrage-free to any individual participant,  yet feature opportunities that arise only through coordinated \emph{exchanges across agents}.

Consider $N$ agents, $i=1,\dots,N,$ each one allowed to invest in some given assets, whose price is represented by a $d_i$-dimensional stochastic process $S^i$. The starting point is the notion of \emph{Collective Arbitrage (CA)}: a zero–sum \emph {risk exchange} among $N$ agents, superposed to their individual trading opportunities, that yields nonnegative terminal payoffs for all agents and a strictly positive payoff on a set of positive probability for at least one of them. The set of allowed risk exchanges among the agents is denoted by $\mcY$ and each vector $Y=(Y^1,\dots,Y^N) \in \mcY $
of random variables represents capital transfers among agents with the aggregate transfer constrained to be zero: $\sum_{i=1}^N Y^i=0$, a.s.. The negation of CA defines \emph{No Collective Arbitrage ($\NCA)$}, a market viability requirement that strictly extends classical no–arbitrage whenever agents face heterogeneous trading sets or information and when risk exchanges are allowed. 
Specifically, it was shown in \cite{BDFFM25} that in discrete time the classical Fundamental Theorem of Asset Pricing (FTAP) must be revisited to account for these interactions: 
$\NCA$ is equivalent to the existence of an \emph{equivalent collective martingale measure} i.e., agent–specific pricing measures $(Q^1,\dots,Q^N)$, each one pricing each agent’s traded assets, that all together render admissible exchanges of at most zero aggregate value: 
\begin{equation}\label{polarity22}
  \sum_{i=1}^N E_{Q^i}[Y^i] \leq 0 \text{ for all } Y \in \mcY,  
\end{equation}
see \eqref{MartingaleLocalMeasures} for details.

Similarly to the classical one-agent case, in the continuous time setting with $N$ agents, the no collective arbitrage assumption is not any more sufficient to guarantee the existence of a collective martingale measure and one has to reinforce the no collective arbitrage condition with the appropriate notion of no collective free lunch ($\NCFL$), see \cite{F25}. It was indeed shown in a general semimartingale setting that $\NCFL$ is equivalent to existence of an \emph{equivalent collective separating measure}, namely a vector of separating measures satisfying the same \textit{polarity} condition \eqref{polarity22}. As is well known, depending on the assumptions imposed on the price process $S$, separating measures take the form of martingale measures, local martingale measures, or 
sigma martingale measures. In this introduction, let us denote, in both discrete time and continuous time, by $\mathcal M(\mcY)$ (resp. $\Me$) the set of such vectors $\mathbf Q=(Q^1,\dots,Q^N)$ of separating probability measures satisfying condition \eqref{polarity22} and such that $Q^i \ll P^i$ (resp $Q^i \sim P^i$  ) for each $i$. \\

A fundamental principle, in the aforementioned reference \cite{BDFFM25}, is that $N$ agents, in addition to trading within their individual markets, can enhance their portfolios through collaborative risk exchanges.
A central contribution of the present paper is the exploration of the relationship between \emph{collective market efficiency and individual rationality}. We formalize an exchange $Y\in \mcY$ as \emph{beneficial} if, given agents’ (concave, increasing) utility functions and subjective probability measures, 
it strictly improves every agent’s indirect utility,  see Definition \ref{defyoptbis} and Remark \ref{remlungo}, Item b).
If, in discrete time, a collective arbitrage exists or if, in continuous time, a collective free lunch exists, then, as intuition also suggests, beneficial exchanges can be found, see Section \ref{BenefitsfromCFL1}.

When a discrete time market is collective-arbitrage-free, (respectively a continuous time model has no collective free lunch), we encounter a more nuanced landscape: a beneficial exchange might or might not exist depending on the specific risk-aversion profiles of the agents in relation to the market structure. 
 This relation between the preferences of the agents and the financial market is mathematically described by the condition that the vector $\mathbf Q_X=(Q_{X^1}^1, \dots,Q_{X^N}^N)$ of \emph{minimax} measures, which are martingale/separating measures, see the Proposition \ref{corgatbis} (and \cite{BelliniFrittelli02} for details), must fulfill the at most zero aggregate price condition $\sum_{i=1}^N E_{Q_{X^i}^i}[Y^i] \leq 0$ that identify, as in \eqref{polarity22}, collective martingale/separating measures. As proven in Theorem \ref{corben} and its Corollary \ref{corbenbis},
 \begin{equation}\label{beneficialexists}
    \text{\emph{beneficial exchanges exist if and only if} }\mathbf Q_X\notin\mathcal M(\mcY)
 \end{equation}
This equivalence pinpoints a precise \emph{compatibility} condition between market properties (as described by $\mathcal M(\mcY)$) and the preferences of the trading participants (determining the minimax measures $\mathbf Q_X$), and clarifies when cooperation strictly improves each agents indirect utility under $\NCA$ or $\NCFL$: unless the preferences are specifically aligned, the condition $\mathbf Q_X\in \mathcal M(\mcY)$ will not hold and thus cooperation will produce exchanges $Y \in \mcY$ that are beneficial for all the agents.

A summary of our findings regarding $\NCA$, $\NCFL$ and existence of beneficial exchanges is provided in the Table \ref{tab:implications} on page \pageref{tab:implications}.

To conclude our Introduction, we mention that there is a whole stream of literature on cooperation in multi-agent systems. Although not immediately comparable, our paper shares some significant conceptual and philosophical points with such literature. Both our framework and the multi-agent systems literature emphasize that cooperation should not be viewed as a deviation from individual rationality, but rather as a mechanism through which agents can jointly achieve outcomes unattainable in isolation. In particular, while our results characterize the existence of beneficial exchanges through a precise compatibility condition between agents’ preferences and collective pricing measures, the multi-agent perspective interprets cooperation as coordinated action that advances either shared or mutually supportive goals among agents.
We mention the paper \cite{DORAN_FRANKLIN_JENNINGS_NORMAN_1997} for a discussion on the multi-agent cooperation. Among foundational contributions in multi-agent systems we mention,  for instance, \cite{Bratman1992} formalizing cooperation as a form of shared intentional activity grounded in joint commitments, \cite{Jennings1993} highlighting the role of commitments and conventions as the basis for coordination among autonomous agents. \cite{DurfeeLesser1988} and \cite{DeckerLesser1995} show how partial global planning and coordination algorithms enable agents to align their actions, and \cite{Steels1990} demonstrates that cooperation may also emerge spontaneously through self-organization, without explicit design. 
In both settings, cooperation arises from the interplay between individual objectives and interaction structures, suggesting a common conceptual foundation in which appropriately aligned incentives, beliefs, or mechanisms transform decentralized behavior into collectively beneficial outcomes.\\

\paragraph{Structure of the paper}

Section \ref{secmainres} presents the formal assumptions, the segmented market model, and the statement of the main theorem, which characterizes the existence of beneficial exchanges in terms of a polarity condition involving the vector of individual agents’ minimax measures.
Section \ref{secharbenefits} develops the theoretical framework and provides the proof of the main result.
Section \ref{sction:theutilmax} reviews the necessary results on the utility maximization problem for a single agent with random endowment in a general semimartingale setting, covering both discrete and continuous time cases.
Section \ref{Benefits} analyzes the emergence of beneficial exchanges in multi-agent markets.
Section \ref{secsettingdiscrete} considers the discrete time framework and related arbitrage concepts, showing that Collective Arbitrages lead to beneficial exchanges. Two alternative proofs are provided: one is based on Theorem \ref{corben} and another provides a more explicit construction.
Section \ref{seccontinuous} focuses instead on the continuous time framework and its implications, showing that Collective Free Lunches lead to beneficial exchanges. As in Section \ref{secsettingdiscrete}, this is first established by directly applying the results from Section \ref{secmainres}, and then complemented by a more constructive proof.
Examples are presented in Section \ref{secexamplesdiscr} (discrete time) and Section \ref{secexamplescont} (continuous time).

\subsection{Assumptions and statement of the main result}
\label{secmainres}

\textbf{The segmented security market.}
We let $(\Omega,\mathcal F,\Pp)$ be a complete probability space and
use the abbreviation a.s.\ to denote $\Pp$-almost sure statements.

We consider $N$ agents $i=1,\dots,N$, $N$ probability measures $P^i$ defined on $(\Omega,\mathcal F)$ with $P^i \sim \Pp$, $N$ $\sigma$-algebras $\mathcal F^i:=\mathcal F^i_T \subseteq \mathcal F$, representing the information at disposal of agent $i$ at some future time $T$, and $N$ vector spaces $$L_i \subseteq L^0(\Omega,\mathcal F^i,\probp^i)=L^0(\Omega,\mathcal F^i,\Pp).$$ Typical examples will be $L_i=L^p(\Omega,\mathcal F^i,\probp^i)$, $p \in (1,\infty]$, but at this stage, we do not need to specify $ L_i$ in
detail; a precise description will be provided in the Assumption \ref{assOK}.
For each $i=1,\dots,N$, we define the \emph{market of agent $i$} as the the convex cone
\[
\mathcal K_i \subseteq L^0(\Omega,\mathcal F^i,\probp^i).
\]

It typically consists of all time-$T$ payoffs that agent~$i$ can generate \emph{at zero initial cost} by trading admissibly in the assets available to him. The price processes of these assets are described by an $\mathbb{R}^{d_i}$-valued general stochastic process $S^i$  defined on a filtered stochastic basis $(\Omega,\mathcal F^i,\mathbb F^i=\{\mathcal F^i_t\}_{t \in \mathcal T},\Pp)$ with $\mathcal T \subseteq [0,T]$, $\mathcal F^i=\mathcal F ^i_T$ and $\mathcal F^i_0$ the trivial sigma-algebra (see Section  \ref{sction:theutilmax} 
).

The $N$ agents~$i$, the domains~$L_i$ and the markets~$\mathcal K_i$ together constitute a \emph{segmented market}. For a detailed motivation and a comprehensive discussion of such segmented market framework, we refer the reader to \cite{Carassus23,DFM25,BDFFM25}.\\

\textbf{The set of allowed exchanges.}
A central premise established in \cite{BDFFM25} is that the group of N agents can optimize their financial positions not only through individual market trading but also by engaging in collaborative risk exchanges.
Let $\mathcal{Y}_0$ denote the set of all zero-sum risk exchanges, defined as
\begin{equation*} \label{2345}
\mathcal{Y}_0 = \left\{ Y=(Y^1,\dots,Y^N)  \mid Y^i \in L^0(\Omega, \mathcal{F}^i, P^i), \,   \sum_{i=1}^N Y^i = 0 \, \text{ a.s.} \right\}
\end{equation*}
and its subset of deterministic exchanges
\begin{equation*} \label{12345}
\R^N_0 = \left\{ x=(x^1,\dots,x^N) \in \R^N  \mid   \sum_{i=1}^N x^i = 0 \, \right\}.
\end{equation*}

While the sum of the components of any $Y \in \mathcal{Y}_0$ is almost surely zero under $\Pp$, the individual components $Y^i$ are, in general, random variables.  A positive realization of $Y^i$ on a particular event signifies a capital inflow for agent $i$ from the collective, while a negative value indicates a capital outflow.  Thus, $Y \in \mathcal{Y}_0$ characterizes the potential capital reallocations among the agents, subject to the constraint of zero net transfer. An example of allowed exchanges is a subset $\mathcal{Y}$ of $\mathcal{Y}_0$, but many other choices are possible. \\

\textbf{Beneficial exchanges.}
We assume that the preferences of the $N$ agents are represented by expected utility, with heterogeneous (concave increasing) utility functions $u^i$ and probability measures $P^i \sim \Pp$, and we suppose that the $N$ agents are expected utility maximizers from investing in their own market $\mathcal K_i$.  Thus for each $i=1,\dots,N $ the indirect utility of agent $i$ is denoted by 
\begin{equation}\label{primal1}
U^i(X^i;Y^i):=\sup_{k^i\in \mathcal{K}_i}E_{P^i}\left[ u^i\left( X^i+k^i+Y^i\right) \right],
\end{equation}%
where $X^i \in L_i$ is the endowment of agent $i$ and $ (Y^i,\dots,Y^N) \in \mcY \subseteq L_1 \times \dots \times L_N$ is an allowed exchange.

\begin{definition}
\label{defyoptbis}
For $(X^1,\dots,X^N)=X \in L_1 \times \dots \times L_N$ and $\widehat{Y} \in  \mcY$ consider, for each $i=1,\dots,N$ the inequality
     \begin{equation} \label{eqIRBis}
        U^i(X^i;\widehat{Y}^i)  \geq U^i(X^i;0).
    \end{equation}
    Given $X$, we say that $\widehat{Y} \in  \mcY$  is \textbf{beneficial} if \eqref{eqIRBis} holds true for all $i=1,\dots,N$, and if for at least one $i \in \{1,\dots,N\}$ the inequality in \eqref{eqIRBis} is strict. We say that $\widehat{Y} \in  \mcY$  is \textbf{strictly beneficial} if the inequality in \eqref{eqIRBis} holds strictly for all $i$.
    \end{definition}
    
When $\widehat{Y}$ is beneficial, the proposed exchanges $\widehat{Y}$ are individually rational \textit{for all agents}. Moreover, every agent's indirect utility, from participating in the exchange is at least as great as their indirect utility would be without it and at least one agent achieves a strict increase in their utility.\\

\begin{remark}
    If we additionally assume that $\mcY\subset \mcY_0$ then the notion of beneficial exchanges coincides with that of Pareto improvements, since the market-clearing condition holds. Nevertheless, we prefer to retain this terminology to emphasize the generality of the set $\mcY$. 
\end{remark}

\textit{The main purpose of this paper is to characterize the existence of beneficial exchanges within the abstract general framework specified in the Standing Assumption and Assumption} \ref{assOK}, \textit{and to interpret this characterization in both discrete and continuous-time security market models.}\\

The Standing Assumptions below are in force throughout the paper.\\

\begin{standass}
\label{assOK2}
We assume the following conditions hold without further mention.
\begin{enumerate}
\item Let $(\Omega, \mathcal{F}, \Pp)$ be a complete probability space,  $P^1, \dots , P^N$, be probability measures, with $P^i \sim \Pp $, and $\mathcal F^1_T, \dots, \mathcal F^N_T $ be $\sigma$-algebras with  $\mathcal F^i:=\mathcal F^i_T \subseteq \mathcal F$, for all $i=1,\dots,N$.

 \item For each $i=1,\dots,N$,
$u^i:\R\rightarrow \R$ is concave, increasing, differentiable on $\R$ and satisfies the Inada conditions. Moreover, $u^i$ is strictly concave on \((-\infty,a)\), where $a:=\inf \{x \in \R : u^i(x)=u^i(+\infty) \} \leq +\infty$.
    \item Let $L_i \subseteq L^0(\Omega,\mathcal F^i,\probp^i)=L^0(\Omega,\mathcal F^i,\Pp)$ be vector subspaces for $i=1,\dots,N$.
    \item 
    The set of exchanges $\mcY$ is a convex cone contained in $L_1 \times \dots \times L_N$.
    \item The random endowments  $X=(X^1,\dots,X^N) \in L_1 \times \dots \times L_N$ are given and satisfy $U^i(X^i;0)<u^i(+\infty)$ for all $i=1,\dots,N$.
\end{enumerate}      
\end{standass}

Theorem~\ref{corben} is the main result of this paper and is formulated under the additional Assumption~\ref{assOK}, introduced below.

\begin{theorem}\label{corben}
Suppose that Assumption \ref{assOK} holds and that  $\R^N_0 \subseteq \mcY$. Then there exists a unique optimizer $Q_{X^i}^i$ of the dual representation of $U^i(X^i,0)$ in \eqref{dual11}, and   
the following are equivalent
 \begin{enumerate}
     \item There exists a strictly beneficial $\widehat{Y} \in \mcY$.
     \item There exists a beneficial $\widehat{Y} \in \mcY$.
     \item The minimax vector $\mathbf Q_X:=(Q_{X^1}^1,\dots,Q_{X^N}^N)$ is not a collective (sigma) martingale\footnote{ In particular $\mathbf Q_{X}$ is not a collective martingale/local martingale/sigma martingale measure, depending on the properties of the price processes or of the market model (see Section 4).} measure for $(S^1, \dots, S^N)$ namely $$\sum_{i=1}^N E_{Q_{X^i}^i}[Y^i]> 0 \text {  for some  } Y \in \mcY.$$  
 \end{enumerate}

\end{theorem}

\begin{remark} (On Theorem \ref{corben}) \label{remlungo}
\begin{enumerate}

    \item[a)] Since for any $X^i \in L_i$ and any $Y^i \in L_i $ we have $U^i(X^i;Y^i) \leq u^i(+\infty) $, it is obvious that a necessary condition for the existence of $Y^i \in L_i$ for which $U^i(X^i;Y^i)>U^i(X^i;0)$ is that $U^i(X^i;0)<u(+\infty)$, which is Item 5 in the Standing Assumption. Such assumption guarantees the existence of a minimax measure $Q^i_{X^i} \ll P^i \sim \Pp$ for $S^i$. However, such measure may not be equivalent to $P^i$, nor to $\Pp$. Conditions under which $Q^i_{X^i} \sim P^i$ and its relation with individual and collective free lunch\slash arbitrage are discussed in Section \ref{Benefits}.

    \item[b)] The equivalence between conditions 1 and 2 is intuitively plausible, given the assumption that $\mathbb{R}^N_0 \subseteq \mathcal{Y}$. Specifically, this assumption enables the redistribution of the additional utility initially generated for the single agent $j$ to all other agents. This ensures that the utility gain can be transferred to achieve a strict improvement, i.e., $U^i(x^i,Y^i) > U^i(x^i,0)$, for all $i$.
\item[c)] The theorem is proved in the abstract framework defined by the Standing Assumption and Assumption \ref{assOK}. The main applications of this result concern financial markets in which $N$ agents, in addition to cooperating, can invest in a given securities market model (either in continuous or in discrete time). \textbf{These applications are illustrated in Section 4 and the main conclusion are formalized in Corollary \ref{corbentris} and Corollary \ref{corbenbis}.}

    \item[d)] 
    We will demonstrate that in continuous time, the presence of a collective free lunch determines the existence of a strictly beneficial exchange (see Proposition \ref{CFLben}) and, similarly, in a discrete-time setting a collective arbitrage implies the existence of such an exchange. 
    \ifLONG
    However, the absence of collective arbitrage does not preclude per se the possibility of beneficial exchanges, as illustrated by the simple explicit example, in discrete time and with finite $\Omega$, in Section \ref{secexamplesdiscr} 
    \fi
    \item[e)] An alternative way to express the condition in Item 3 of Theorem \ref{corben} is to state that the vector $\mathbf Q_X$ does not belong to the set of collective (sigma) martingale measures  $\mathcal{M}(\mcY)$,  as this set is characterized by the polarity condition \eqref{polarity22}. Let us discuss this necessary and sufficient condition, namely $\mathbf Q_X \notin \mathcal{M}(\mcY)$, for the existence of beneficial exchanges. As it will be clarified in the next sections, it is crucial to note that the set $\mathcal{M}_{e}(\mcY) := \{\mathbf Q=(Q^1,\dots,Q^N) \mid Q^i \sim P^i \text{ and }  Q \in  \mathcal{M}(\mcY) \}$ is determined solely by collective arbitrage considerations (and, naturally, by the chosen admissible set $\mathcal{Y}$ of exchanges), whereas $\mathbf Q_X$ is primarily derived from the agents' preferences (and the selection of the markets $\mathcal K_i$). Consequently, when the minimax vector $\mathbf Q_X:=(Q_{X^1}^1,\dots,Q_{X^N}^N)$ has all components $Q^i_{X^i}$ equivalent to $P^i$, then the condition $\mathbf Q_X \in \mathcal{M}_{e}(\mcY)$ represents a \textit{compatibility condition between collective arbitrage and the agents' preferences}. In general, unless agents' preferences are specifically aligned, this condition will not hold, which then opens up the possibility that certain $Y \in \mathcal{Y}$ are strictly beneficial also under $\mathbf{NCFL(\mcY)}$ or $\mathbf{NCA(\mcY)}$. 
    
    \item[f)] Theorem \ref{corben} establishes the equivalence in the context of heterogeneous utility functions $u^i$,  probability measures $P^i \sim P$, and of possibly distinct filtrations $\mathbb F^i$. \textit{We do not exclude the possibility that all agents have access to the same market}, that is, $\mathcal{K}_1 = \dots = \mathcal{K}_N$. In this case, the minimax measures $Q^i_{X^i}$ appearing in item 3 will possibly differ from one another only because the utility functions $u^i$ and/or the probability measures $P^i$ and/or the initial endowment $X^i$ differ.

    \item[g)]The density $\rn{Q_{X^i}^i}{P^i}$ of the unique optimizer $Q_{X^i}^i$ of the dual representation of $U^i(X^i,0)$ given in \eqref{dual11} belongs to the dual space of $L_i$ and thus the expectation $E_{Q_{X^i}^i}[Y^i]$ is well defined for all $Y \in \mcY \subseteq L_1 \times \dots \times L_N$.

    \item [h)] 
    The convex cone $\mathcal Y$ is only required to satisfy $\mathbb R^N_0 \subseteq \mathcal Y \subseteq L_1 \times \dots \times L_N$, and it is clear that the smaller the set $\mathcal Y$, the more restrictive the existence of beneficial exchanges becomes. The extreme case corresponds to $\mathcal Y=\mathbb R^N_0$. As intuition suggests, in this situation no beneficial exchange can exist: indeed, for any $x\in\mathbb R^N_0=\mcY$ with $x\neq 0$, there is at least one component $x^j<0$, implying that agent $j$ necessarily incurs a loss and therefore cannot strictly increase their utility.
    This intuition is confirmed by Theorem~\ref{corben}. Indeed, regardless of the choice of the minimax vector $Q_X$, the condition in item~(3) of the theorem is never satisfied, since for every $x\in\mathbb R^N_0$, $\sum_{i=1}^N \mathbb E_{Q_{X^i}^i}[x^i] = \sum_{i=1}^N x^i = 0.$

    \item [g)]
    On the Pareto optimality of beneficial exchanges.
The indirect utility \(U^i(X^i,\cdot):L_i\to\mathbb{R}\) naturally induces an individual preference relation \(\succeq_i\). In our setting, it is therefore natural to investigate the notion of Pareto optimality.
\\ An allocation \(Y\in\mathcal{Y}\) is said to be \emph{Pareto optimal} if there exists no other allocation \(\widetilde{Y}\in\mathcal{Y}\) such that \(\widetilde{Y}^i \succeq_i Y^i\) for every \(i=1,\dots,N\), and \(\widetilde{Y}^j \succ_j Y^j\) for at least one \(j\).
\\ As in standard equilibrium theory, for any \(\lambda_1,\dots,\lambda_N>0\), every optimizer \(\widehat{Y}\in\mathcal{Y}\) of the problem
\[
\sup_{Y\in\mathcal{Y}} \sum_{i=1}^N \lambda_i U^i(X^i,Y^i)
\]
is Pareto optimal.
\\ Nevertheless, establishing the existence of a Pareto-optimal allocation that is also beneficial remains a challenging question. The main difficulty in proving this conjecture lies in the fact that the upper-level sets
\[
\left\{(Y^1,\dots,Y^N)\in\mathcal{Y}\;\middle|\; U^i(X^i;Y^i)\ge U^i(X^i;0)\ \forall i=1,\dots,N\right\}
\]
are closed with respect to the product topology \(\sigma(L_1,L_1^*)\times \ldots\times \sigma(L_N,L_N^*) \), while compactness cannot, in general, be guaranteed.

\end{enumerate}
    
\end{remark}

One simple application of Theorem \ref{corben} is the case when each individual market is reduced to the zero element  (e.g. when $\mathcal K_i=\{0\}$ for all $i$), and the initial endowment is not random. As proven in Corollary \ref{prop3conditions}, in this case the existence of beneficial exchanges is possible if and only if the subjective probabilities $(P^1,\dots,P^N)$ do not satisfy the condition $\sum_{i=1}^N E_{P^i}[Y^i] \leq  0 \text {   }\forall Y \in \mcY$. In particular, existence of beneficial exchanges does not depend any more on the utility functions $u^i$ of the agents, but only on their subjective probability $P^i$.
If, for all $i$, $\mathcal K_i=\{0\}$ and the initial endowments $X^i=x^i \in \R$ are non random, one readily checks that: (i) the vector of minimax measures coincides with the vector of subjective probabilities $(P^1,\dots,P^N)$; (ii) the indirect utility simplifies to the expected utility. Thus from Theorem \ref{corben} we deduce

\begin{corollary}\label{prop3conditions}
Suppose that Assumption \ref{assOK} holds true, that  $\R^N_0 \subseteq \mcY$, and let $X=x=(x^1,\dots,x^N) \in \R^N$, which satisfies $u^i(x^i)<u^i(+\infty)$ for all $i=1,\dots,N$. If $K_i=\{0\}$ for all $i=1,\dots,N$ then
the following are equivalent
 \begin{enumerate}
     \item There exists $Y \in \mcY$ such that $E_{P^i}[u^i(x^i+Y^i)] > u^i(x^i) $ for all $i \in \{1,\dots,N\}$.
     \item There exists $Y \in \mcY$ such that $E_{P^i}[u^i(x^i+Y^i)] \geq  u^i(x^i) $ for all $i \in \{1,\dots,N\}$ and $E_{P^j}[u^j(x^j+Y^j)] > u^j(x^j) $ for some $j \in \{1,\dots,N\}$.
     \item  $\sum_{i=1}^N E_{P^i}[Y^i]> 0 \text {  for some  } Y \in \mcY$. 
 \end{enumerate}
\end{corollary}

As already stated, throughout the paper, we work under the Standing Assumptions. Theorem \ref{corben} and its Corollary are established under the following further assumptions on the ambient spaces $L_i$ and the  set $\mathcal K_i$ of terminal payoffs from admissible trading strategies 

\begin{assumption}
\label{assOK}
 For all $i=1,\dots,N$ 

\begin{enumerate}

    \item $L_i \subseteq L^0(\Omega,\mathcal{F}^i,P^i)$
     is a solid\footnote{$Y\in L_i$, $X\in L^0(\Omega,\mathcal{F}^i,P^i)$ with $\abs{X}\leq \abs{Y}$ a.s. implies $X\in L_i$.} Banach lattice (under the a.s. ordering)  with $L^\infty(\Omega,\mathcal{F}^i,P^i)\subseteq L_i$ and such that $E_{P^i}[u^i(X)]>-\infty$ for every $X\in L_i$.
     \item The K\"{o}the dual of $L_i$ is the space $L_i^*$ given by
\begin{equation}\label{K*}
    L_i^*:=\{Z\in L^0(\Omega,\mathcal{F}^i,P^i)\text{ s.t. } E_{P^i}[{\abs{XZ}}]<+\infty\,\forall X\in L_i\}
\end{equation}
and we adopt the dual pairing $(L_i,L^*_i)$.
    \item  $\mathcal K_i $ is a convex cone contained in $L^0(\Omega,\mathcal{F}^i,P^i)$ and there exists a random variable $W_i \in L_i$, $W_i \geq 1$ a.s., such that for every $k\in \mathcal K_i$ there exists a constant $\alpha\geq 0$ such that $k\geq -\alpha W_i$ a.s..
    
\end{enumerate}  
\end{assumption}

Assumption \ref{assOK} item 1 implies that for $\widehat{u}^i(x):=u^i(0)-u^i(-\abs{x})$ it holds: $E_{P^i}[\widehat{u}^i(\lambda \abs{X})]<+\infty$ for every $\lambda>0$ and $X \in L_i$, that is $L_i\subseteq  M^{\widehat{u}^i}(\Omega,\mathcal{F}^i,\probp^i)$ the latter being the Orlicz Heart associated to the Young function $\widehat{u}^i$, see the Appendix \ref{approlicz} and \cite{BF08} for details. It is well known that Orlicz Hearts are included in  $ L^1$ spaces, so that each ambient space $L_i$ satisfies $$L^\infty(\Omega,\mathcal{F}^i,\probp^i)\subseteq L_i \subseteq M^{\widehat{u}^i}(\Omega,\mathcal{F}^i,\probp^i) \subseteq L^1(\Omega,\mathcal{F}^i,\probp^i). $$
This also shows that the expectations in Item 3 of Corollary \ref{prop3conditions} are well defined, as $\mcY \subseteq L_1 \times \dots \times L_N \subseteq L^1(\Omega,\mathcal{F}^1,\probp^1) \times \dots \times L^1(\Omega,\mathcal{F}^N,\probp^N)$.

\begin{remark}
Let $p\in(1,+\infty)$. For the choice
\[
u^i(x):=-|x|^p\,\mathbf 1_{(-\infty,0]}(x)
\quad\text{and}\quad
L_i=M^{\widehat{u}^i}(\Omega,\mathcal F^i,\probp^i),
\]
one recovers
\[
L_i=L^{p}(\Omega,\mathcal F^i,\probp^i).
\]
While our assumptions clearly allow for the choice
$L_i=L^\infty(\Omega,\mathcal F^i,\Pp)=L^\infty(\Omega,\mathcal F^i,P^i)$, they exclude the case
$L_i=L^1(\Omega,\mathcal F^i,\probp^i)$ on any probability space supporting
unbounded random variables (for instance, any non-atomic probability space).
Indeed, the Inada conditions imply that
$
\frac{\widehat{u}^i(y)}{y}\longrightarrow +\infty
\quad\text{as } y\to +\infty,
$
which in turn yields the strict inclusion
$
M^{\widehat{u}^i}(\Omega,\mathcal F^i,\probp^i)
\subsetneq
L^1(\Omega,\mathcal F^i,\probp^i).
$
\end{remark}

In case $L_i=M^{\widehat{u}}(\Omega,\mathcal{F}^i,\probp^i)$, then the K\"{o}the dual is $L_i^*=L^{\widehat{\Phi}^i}(\Omega,\mathcal{F}^i,\probp^i)$, the Orlicz space associated to the convex conjugate $\widehat{\Phi}^i$ of $\widehat{u}^i$, given by $\widehat{\Phi}^i(y)=\sup_{x \in \R}(xy-\widehat{u}^i(x))$;  in case $L_i=L^{\infty}(\Omega,\mathcal{F}^i,\Pp)=L^{\infty}(\Omega,\mathcal{F}^i,\probp^i)$ then the K\"{o}the dual is $L_i^*=L^{1}(\Omega,\mathcal{F}^i,\probp^i)$. Employing the dual pair $(L, L^*)$ not only provides the flexibility to accommodate a broader range of function spaces but also offers the significant advantage that the dual space of $L_i=L^\infty(\Omega,\mathcal{F}^i,P^i)$ or $L_i=M^{\widehat{u}^i}(\Omega,\mathcal{F}^i,\probp^i)$ is a subspace of integrable random variables.

\begin{remark}
In some situations, it is natural to take $L_i=L^{\infty}(\Omega,\mathcal{F}^i,P^i)$. 
However, as discussed in detail in \cite{BF05}, this choice becomes too restrictive when dealing with utility maximization problems in models, including discrete time ones, where the price process $S$ is unbounded, see also Remark \ref{remW}. 
For this reason, the appropriate functional framework for utility optimization is provided by Orlicz Spaces and, in particular, by the corresponding Orlicz Heart; see the Appendix and \cite{BF08}.
\end{remark}

\section{Proof of the main result}
\label{secharbenefits}

We recall that throughout the paper  the Standing Assumptions hold true. We remark that the Standing Assumption 5 is not required in Lemma \ref{l2} and Proposition \ref{propU*}.

\begin{notation}
Given the complete probability space  $(\Omega, \mathcal{F}, \Pp)$,
we say that a probability measure $Q$ on $(\Omega, \mathcal{F})$ belongs to $\mathcal{P}_e$ if $Q \sim \Pp$, or to $\mathcal{P}_{ac}$ if $Q\ll \Pp$, respectively.
 Unless differently stated, all inequalities between random variables are meant to hold a.s., that is $\Pp$-a.s.. 
 We will denote the expectation of a random variable under a probability $P$
by $E_P[\cdot ]$. 
If $L$, $K\subseteq L^{0}(\Omega, \mathcal{F} , \Pp)$, we set: $L_{+}:=\left\{ f\in L \mid f\geq 0\right\} ;$
$L-K:=\left\{ f-k \mid f\in L\text{ and }k\in K\right\} .$ For $f \in (L^{0}(\Omega, \mathcal{F} , \Pp))^N$ we set $f^+:=((f^1)^+, \dots,(f^N)^+)$ where $ (f^j)^+:=\max \{f^j,0 \} $ and $f^-:=((f^1)^-, \dots,(f^N)^-)$ where $ (f^j)^-:=\max \{-f^j,0 \} $.

We denote with
$\Phi^i(y):=\sup_{x\in\mathbb{R}}\big(u^i(x)-xy\big)=-(u^i)^*(y)$, $y \in \R, $ the convex conjugate of $u^i$, which is bounded from below  and strictly convex on $(0,+\infty) $.

\end{notation}

Given the convex cone $\mathcal K_i \subseteq L^0(\Omega, \mathcal F^i,P^i)$, consider the convex cone 
\begin{equation}
    \label{Ci}
 \mathcal{C}_i=(\mathcal K_i -L^0_+(\Omega, \mathcal F^i,P^i) \cap L_i \subseteq L_i   
\end{equation}
and denote by $$\mathcal{C}_i^0=\{\xi \in L_i^* \mid E_{P^i}[\xi f] \leq 0 \text{ for all } f \in \mathcal C_i \} $$ the polar of $\mathcal{C}_i$ in the dual pair $(L_i,L_i^*)$.

 \begin{lemma}
\label{l2}
For every $g\in (L_i)_+$ and $k\in\mathcal{K}_i$, $\min(k,g)\in \mathcal{C}_i.$ Moreover, if Assumption \ref{assOK} holds, for every $X\in L_i$
\begin{equation}
\sup_{k\in \mathcal{K}_i}E_{P^i}[u^i(X+k)]=\sup_{f\in \mathcal{C}_i}E_{P^i}[u^i(X+f)].
\label{88}
\end{equation}
\end{lemma}
\begin{proof}
 For $g\in (L_i)_+ $ and $k \in \mathcal{K}_i$ we have $\min (k,g)=k-(k-g)1_{\{k\geq g\}}\in \mathcal{K}_i- L^0_+(\Omega, \mathcal F^i,P^i)$. We have that  $(\min(k,g))^+\leq g^+\in L_i$. Hence $(\min(k,g))^+\in L_i$. As for negative parts $$\min(k,g)1_{\{\min(k,g)\leq 0\}}=f1_{E}+g1_F,\quad E:=\{f\leq 0, f\leq g\},\, F:=\{g\leq 0, g\leq f\}.$$
Hence $(\min(k,g))^-= f^-1_E+g^-1_F$. We observe that $\abs{f^-1_E}\leq \abs{f^-}=f^-\in L_i$, and $\abs{g^-1_F}\leq g^-\in L_i$ similarly. Hence also $(\min(k,g))^-\in L_i$. We conclude that $(\min(k,g))=(\min(k,g))^+-(\min(k,g))^-\in L_i$.

We now move on to prove \eqref{88}. We see that  $X-k^- \in L_i$ for every $k\in \mathcal{K}_i$ by the assumptions on $\mathcal{K}_i$, so that in particular $u^i(X-k^-)\in  L^1(\Omega, \mathcal{F}^i,\probp^i)$.
It is also clear that $\mathcal{C}_i\subseteq \widetilde{\mathcal{C}}_{i,X}:=\{f\in \mathcal{K}_i-L_{+}^{0}(\Omega, \mathcal{F}^i,\probp^i)\text{ s.t. }u^i(X-f^-)\in L^1(\Omega, \mathcal{F}^i,\probp^i)\}$.
For any $f=k-\ell_+\in \widetilde{\mathcal{C}}_{i,X}$ we have that $\Epo{u^i(X+f)}$ is well defined (as $u^i(X+f)\geq u^i(X-f^-)\in L^1(\Omega, \mathcal{F}^i,\probp^i)$), and so is $\Epo{u^i(X+k)}$ (as $u^i(X+k)\geq u^i(X-k^-)\in  L^1(\Omega, \mathcal{F}^i,\probp^i)$), with $\Epo{u^i(X+f)}\leq \Epo{u^i(X+k)}$.  Hence
$$\sup_{f\in \mathcal{C}_i}E_{P^i}[u^i(X+f)]\leq \sup_{f\in \widetilde{\mathcal{C}}_{i,X}}E_{P^i}[u^i(X+f)]\leq \sup_{k\in \mathcal{K}_i}E_{P^i}[u^i(X+k)]$$
using monotonicity of $u^i$ in the last inequality.
Furthermore, for every $k\in \mathcal K_i$ and $n\in\mathbb{N}$ we have $k_n:=\min(k,n)=k-(k-n)1_{\{k\geq n\}}\in \mathcal C_i$ (by what was first showed in the proof), and $u^i(X-k^-)\leq u^i(X+k_n)\uparrow u^i(X+k)$. By Monotone Convergence Theorem ($u^i(X-k^-)\in L^1(\Omega, \mathcal{F}^i,\probp^i)$) then $\Epo{u^i(X+k)}=\lim_n\Epo{u^i(X+k_n)}\leq \sup_{f\in \mathcal{C}_i}E_{P}[u^i(X+f)]$. This shows that $\sup_{k\in \mathcal{K}_i}E_{P^i}[u^i(X+k)]\leq \sup_{f\in \mathcal{C}_i}E_{P^i}[u^i(X+f)]$.
\end{proof}

The K\"{o}the dual $L_i^*$ of $L_i$ is defined in \eqref{K*} and we denote with $L_i'$ the topological dual space of the Banach lattice $L_i$.

\begin{remark}
    $L_i^*$ can be obviously seen as a subspace of the algebraic dual of $L_i$ in that every $Z\in L_i^*$ induces a linear functional $\phi_Z$ on $L_i$ via
$\phi_Z(X):=\Epo{ZX}, X\in L_i$. Actually, these functionals turn out to be norm continuous: the K\"{o}the dual coincides with all and only the functionals that are $\sigma$-order continuous (and in particular continuous in norm) over $L_i$. From Assumption \ref{assOK} it also follows that $L^\infty(\Omega,\mathcal{F}^i,\probp^i)\subseteq L_i^*\subseteq L^1(\Omega,\mathcal{F}^i,\probp^i)$, and that $L_i,L_i^*$ are decomposable (under $P^i$) in the sense of \cite{Rockafellar} page 532. All the claims above are made precise, and references provided, in Lemma \ref{kotheprop}. More details can be found in \cite{Kanto} Chapter IV.
\end{remark}

Define $U^i:L_i \rightarrow [-\infty,+\infty]$ by
\begin{equation}\label{eqUW}
U^i(X):=\sup_{k\in\mathcal{C}_i}E_{P^i}[u^i(X+k)]=U^i(X;0).
\end{equation}

\begin{proposition}\label{propU*}
 Suppose Assumption \ref{assOK} holds. Assume
 that $U^i(\bar{X})<+\infty$ for some $\bar{X}\in L_i$.
 Then $\mathrm{dom}(U^i)=:\{X \in L_i \mid U^i(X)>-\infty \}=L_i$ and
$U^i :L_i \rightarrow \R$ is concave, nondecreasing, norm continuous on $L_i$. Moreover
\begin{equation} \label{dualrho1bis}
U^i(X) =\min_{\xi\in L_i'}
\Big(\xi(X)-(U^i)^*(\xi)\Big)    
\end{equation}
where for every $\xi\in L_i'$ it holds that\footnote{We note that $\xi\in \mathcal{C}_i^0$  already implies that $\xi \geq 0$, as $-L^\infty_+(\Omega, \mathcal F ^i,P^i)\subseteq \mathcal{C}_i$. However, we prefer to stress this positivity explicitly for later use.}
\begin{equation}
    \label{conjugatebis}
    (U^i)^*(\xi):=\inf_{X\in L_i}\left(\xi(X)-U^i(X)\right)=\begin{cases}
    E_{P^i}\big[(u^i)^* (\xi)\big]&\text{ if } \xi\in \mathcal{C}_i^{0}, \xi\geq 0\\
    -\infty&\text{ otherwise}.
\end{cases}
\end{equation}
Finally, 
\begin{equation}\label{eqsuperdiff}    
\partial U^i(X)=\{\xi \in  \mathcal{C}_i^{0} \mid \xi\text{ is optimum for RHS of } \eqref{dualrho1bis} \} \neq \emptyset\quad \forall X\in L_i,
\end{equation}
    where $\partial U^i(X) $ is the superdifferential of $U^i$ at $X$.    
\end{proposition}

\begin{proof}
Since $0\in\mathcal{C}_i$ we have $-\infty<\Epo{u^i(X)}\leq U^i(X)$ for every $X\in L_i$. Furthermore an elementary verification yields:
$U^i(\lambda X_1+ (1-\lambda X_2)\geq \lambda U^i(X_1)+ (1-\lambda )U^i(X_2)$ for every $X_1, X_2\in L_i$ and $0<\lambda<1$, with usual conventions for operations with $+\infty$. Since $U^i(\overline{X})$ is finite by assumption, it then follows that $U^i(X)$ is finite for every $X\in L_i$, and concave. One also verifies monotonicity of $U^i$ directly by monotonicity of $u^i$. 
 By the Extended Namioka–Klee Theorem
(see \cite{bfnam}, Theorem 1) we have that $U^i$ is norm continuous on $\mathrm{dom}(U^i)=L_i$ and
\begin{equation*}
    \label{dualrho1}
U^i(X)
=
\min_{\xi\in {L_i'}}
\Big(\xi(X)-(U^i)^*(\xi)\Big),    
\end{equation*}
where for $\xi \in L_i'$
\begin{align}
(U^i)^*(\xi)
&:=
\inf_{X\in L_i}\Big(\xi(X)-U^i(X)\Big) \notag
=
\inf_{X\in L_i}
\Big(
\xi(X)
-\sup_{k\in\mathcal{C}_i}E_{P}\big[u^i(X+k)\big]
\Big) \notag
\\
&=
\inf_{k\in\mathcal{C}_i}
\inf_{X\in L_i}
\Big(
\xi(X)-E_{P^i}\big[u^i(X+k)\big]
\Big) =
\inf_{k\in\mathcal{C}_i}
-\xi(k)
+
\inf_{Y\in L_i}
\big(
\xi(Y)-E_{P^i}[u^i(Y)]
\big) \label{U^i*bis}
\end{align}
and in \eqref{U^i*bis} we made the change of variables $Y:=X+k \in L_i$.
Notice that in the definition of $(U^i)^*$ we are using the dual system $(L_i,L_i')$, where $L_i'$ is the topological dual space of the Banach lattice $L_i$. The polar cone $\mathcal{C}_i^{0}$
is instead, by our definition, computed with respect the dual system $(L_i,L_i^*)$.


a) We first claim that if $\xi \in L_i'$ and $(U^i)^*(\xi)>-\infty$ then $\xi(f)\le 0$ for all
$f\in\mathcal{C}_i$.
Indeed, 
since  $\inf_{Y\in L_i}
(
\xi(Y)-E_{P^i}[u^i(Y)])\leq u^i(0)<+\infty$, we deduce from \eqref{U^i*bis} that $\inf_{k\in\mathcal{C}_i}
-\xi(k)>-\infty$. As  $0\in\mathcal{C}_i$ and $\mathcal{C}_i$ is a cone, we get $\inf_{k\in\mathcal{C}_i}
-\xi(k)=0$ that is our claim. \\ 

b) Our second claim is that if $\xi \in L_i'$ and $(U^i)^*(\xi)>-\infty$ then $\xi \in L_i^*$.\\
We momentarily suppose that b) holds and conclude the proof as follow.
Claims a) and b) imply that if $\xi \in L_i'$ and $(U^i)^*(\xi)>-\infty$ then $\xi \in \mathcal C_i^0$. Therefore, if $\xi \in L_i'$, but $\xi \notin \mathcal C_i^0$, then  $(U^i)^*(\xi)=-\infty$. 
Thus, to end the proof of \eqref{conjugatebis}, it remains to prove that for $\xi\in \mathcal{C}_i^0$ we have $(U^i)^*(\xi)=\Epo{(u^i)^*(\xi)}$.
We note that the Assumption \ref{assOK} and Lemma \ref{kotheprop} Item 1 ensure that the assumptions of \cite{Ro89}, Corollary on page 534 are satisfied. Thus the concave conjugate $(I_{u^i})^*:L_i^* \rightarrow [-\infty,+\infty)$ of the integral functional $I_{u^i}(X)=E_{P^i}[u^i(X)]$ can be represented as \begin{equation}\label{I*}
(I_{u^i})^*(\xi):=\inf_{Y\in L_i}
\big(
\xi(Y)
-E_{P^i}[u^i(Y)]
\big)=E_{P^i}\big[(u^i)^*(\xi)\big)], \quad \xi \in L_i^*,
\end{equation}
so this latter expression can be inserted in \eqref{U^i*bis}, which gives \eqref{conjugatebis} recalling that $\inf_{k\in\mathcal{C}_i}
-\xi(k)=0$ if $\xi\in \mathcal{C}_i^0$.

To prove b) we now show that if  $\xi \in L_i'$ with $(U^i)^*(\xi)>-\infty$ then $\xi$ is $\sigma$-order continuous in the sense of Lemma \ref{kotheprop} Item 2-3, thus is represented by an element in $L_i^*$ (Lemma \ref{kotheprop} Item 2).
By definition $-(L_i)_+ \subseteq \mathcal{C}_i$ and, by claim a), $\xi(f)\leq 0$ for $f\in\mathcal{C}_i$. Thus   $\xi(X)\geq 0$ for every $X\in L_i, X\geq 0$.
 From \eqref{U^i*bis}, since $\inf_{k\in\mathcal{C}_i}
-\xi(k)=0$,
\begin{equation}
\label{coniug1bis}
    (U^i)^*(\xi)=\inf_{Y\in L_i}
(
\xi(Y)-E_{P^i}[u^i(Y)]).
\end{equation}
Take a sequence $0\leq X_n\rightarrow_n0$ with $\sup_n \abs{X_n}\in L_i$ and $\lambda >0$. Replacing $Y$ with $-\lambda X_n$ in \eqref{coniug1bis} we see that
$$0\leq \lambda \xi(X_n)\leq -\Epo{u^i(-\lambda X_n)}-(U^i)^*(\xi).$$
Dominated Convergence Theorem ensures $\Epo{u^i(-\lambda X_n)}\rightarrow u^i(0)$, so that $0\leq \limsup_n \xi(X_n)\leq \frac1\lambda(-u^i(0)-(U^i)^*(\xi))$. Taking the infimum over $\lambda >0$ in RHS yields $\lim_n \xi(X_n)=0$. This proves that $\xi$ is $\sigma$-order continuous (Lemma \ref{kotheprop} Item 3), 

The final claim of the Proposition follows from the  standard fact that the subdifferential of a convex function at a given point coincides with the set of optima for its dual representation at that given point.
\end{proof}

\begin{proposition}
\label{corgatbis}
Suppose Assumption \ref{assOK} holds.
    Fix $X\in L_i$ and suppose that 
    $$U^i(X)=U^i(X;0)=\sup_{f\in \mathcal{C}_i}E_{P}\big[u^i(X+f)\big]<u^i(+\infty).$$
   Then there exists a unique pair $(\lambda_X,\probq_X) \in \R \times \mathcal P_{ac}$,  with $\lambda_X>0$ and   such that $\rn{Q_X}{P} \in \mathcal{C}_i^{0}$,  which satisfies 
    \begin{equation}\label{dual11}
    \begin{split}
    U^i(X;0)=\sup_{f\in\mathcal{C}_i}E_{P^i}[u^i(X+f)]&=\min_{\lambda >0} \min_{\substack{  \probq\ll\probp^i\\ \rn{Q}{P^i} \in \mathcal{C}_i^{0}}}
\Bigg (\lambda E_{Q}[X]+E_{P^i}\bigg[\Phi^i\Big(\lambda \rn{\probq}{\probp^i}\Big)\bigg] \Bigg )\\  
&=
\lambda_XE_{Q_X}[X]+E_{P^i}\bigg[\Phi^i\Big(\lambda_X \rn{\probq_X}{\probp^i}\Big)\bigg].
    \end{split}
\end{equation}
The measure $Q_X$ is called \textbf{minimax measure}. 
Furthermore, for every $Y\in L_i$ the function
$$F^i_Y(\alpha):=\sup_{f\in\mathcal{C}_i}E_{P^i}[u^i(X+\alpha Y+f)]=U^i(X+\alpha Y)$$
is well defined in a neighborhood of zero, and differentiable at zero with
$$\rn{}{\alpha}F^i_Y\Big |_{\alpha=0}=\lambda_X\,E_{\probq_X}[Y].$$
\end{proposition}

\begin{proof}
Notice that $U^i(X)<u^i(+\infty)$, so that all the assumptions of Proposition \ref{propU*} are satisfied. Any $0\leq \xi \in \mathcal{C}_i^0$, $\xi \neq 0$ can be represented by $\xi =\lambda \rn{Q}{P^i}$, where $\lambda > 0$ and the probability $Q \ll P^i$ satisfies $\rn{Q}{P^i} \in \mathcal{C}_i^{0}$.  We also note that $U^i(X)<u^i(+\infty)=\Phi^i(0)$ implies that the minimum in RHS of \eqref{dualrho1bis} cannot be attained at $\xi=0\in\mathcal{C}_i^0$. As $\Phi^i=-(u^i)^*$ from \eqref{dualrho1bis} and \eqref{conjugatebis} we deduce 
\begin{equation} \label{dualrho1tris}
\begin{split}
U^i(X) =
\min_{\xi\in \mathcal{C}_i^{0},\xi\neq 0}
\Big(\xi(X)-E_{P^i}\big[(u^i)^* (\xi)\big]\Big)
&= \min_{\lambda >0} \min_{\substack{  \probq\ll\probp^i\\ \rn{Q}{P^i} \in \mathcal{C}_i^{0}}}
\Bigg (\lambda E_{Q}[X]+E_{P^i}\bigg[\Phi^i\Big(\lambda \rn{\probq}{\probp^i}\Big)\bigg] \Bigg )  \\
&= \lambda_XE_{Q_X}[X] +E_{P^i}\bigg[\Phi^i\Big(\lambda_X \rn{\probq_X}{\probp^i}\Big)\bigg].
\end{split}
\end{equation}

By Lemma \ref{lemmauphi}, the function $\Phi^i$ is strictly convex on $(0,+\infty)$, so that the minimizers $\lambda_X$, $\probq_X$ in \eqref{dualrho1tris} are uniquely determined.
    Uniqueness of such optima implies, by \eqref{eqsuperdiff}, that $\partial U^i(X)$ is a singleton. Then, $U^i$ it is Gateaux differentiable at $X$
    by \cite{BoVa10} Proposition 6.1.4.(a), and  $F^i_Y(\alpha)=U^i(X+\alpha Y)$ 
is differentiable at $0$ by definition of Gateaux differentiability.
Since $U^i$ is norm continuous, $U^i(X+\alpha Y)<u^i(+\infty)$ for $\alpha$ sufficiently small. We apply \eqref{dualrho1tris} replacing there $X$ with $(X+\alpha Y)$ and deduce, for $\alpha$ sufficiently small,
\begin{align}\label{dualdual}
F^i_Y(\alpha)&=\sup_{f\in\mathcal{C}_i}E_{P^i}[u^i(X+\alpha Y+f)]=U^i(X+\alpha Y) \notag \\
&= \min_{\lambda >0} \min_{\substack{  \probq\ll\probp^i\\ \rn{Q}{P^i} \in \mathcal{C}_i^{0}}}\bigg(\lambda \big(E_{Q}[X]+\alpha E_{\probq}[Y]\big)+E_{P^i}\bigg[\Phi^i\Big(\lambda \rn{\probq}{\probp^i}\Big)\bigg]\bigg):=\min_{\lambda >0} \min_{\substack{  \probq\ll\probp^i\\ \rn{Q}{P^i} \in  \mathcal{C}_i^{0}}} g^i_{\lambda,Q}(\alpha).
\end{align}
By the envelope Theorem (\cite{MiSe02}, Theorem 1) we have $\rn{}{\alpha}F^i_{Y}\Big |_{\alpha=0}=\rn{}{\alpha}g^i_{\lambda_X,Q_X}\Big |_{\alpha=0}=\lambda_X\,E_{\probq_X}[Y]$.
\end{proof}
\subsection{Proof of Theorem \ref{corben}}

\begin{proof}
We prove that: 1. $\Rightarrow $ 2. $\Rightarrow$ 3.  $\Rightarrow $ 1. The implication 
1. $\Rightarrow $ 2. is obvious.\\
Suppose 2. holds and let $Y \in \mcY$  satisfy the inequalities in item 2.  Set $F^i_{Y^i}(\alpha):=U^i(X^i;\alpha Y^i)=U^i(X^i+\alpha Y^i)$, $\alpha \in \R$, and note that $F^i_{Y^i}(1)=U^i(X^i;Y^i)$ and $F^i_{Y^i}(0)=U^i(X^i;0)$.
By Fenchel inequality
\begin{equation*}
    F^i_{Y^i}(\alpha) \leq \inf_{\lambda >0} \inf_{\substack{  \probq\ll\probp\\ \rn{Q}{P} \in (\mathcal{C}_i)^{0}}}
\bigg(
\lambda(E_Q[X^i] + \alpha  E_{Q}[Y^i])
+ E_{P^i}\left[\Phi^i\left(\lambda\frac{\mathrm{d}{Q}}{\mathrm{d}P^i}\right)\right]\bigg),
\end{equation*}
and by \eqref{dual11} 
\begin{equation*}
   F^i_{Y^i}(0)=\bigg(\lambda_{X^i} E_{Q^i_{X^i}}[X^i]  + E_{P^i}\left[\Phi^i\left(\lambda_{X^i}\rn{Q^i_{X^i}}{P^i}\right)\right]\bigg).
\end{equation*}
for unique pairs$(\lambda_{X^i},\probq^i_{X^i})$ with $\lambda_{X^i}>0$,  $Q^i_{X^i}  \ll P^i$ satisfying $\rn{Q^i_{X^i} }{P^i} \in (\mathcal{C}_i)^{0}$.
Thus computing the right derivative of $F^i_{Y^i}$ at zero we get
\begin{align*}
   \frac{\mathrm{d}^+}{\mathrm{d}\alpha} F^i_{Y^i} \Big |_{\alpha=0}&= \lim_{\alpha \downarrow 0}\frac1\alpha\bigg(F^i_{Y^i}(\alpha)-F^i_{Y^i}(0)\bigg)\\& \leq \lim_{\alpha \downarrow 0}\frac1\alpha\bigg \{ \inf_{\substack{ \lambda>0, \probq\ll\probp\\ \rn{Q}{P} \in (\mathcal{C}_i)^{0}}}
\bigg(
\lambda(E_{Q}[X^i] + \alpha  E_{Q}[Y^i])
+ E_{P^i}\left[\Phi^i\left(\lambda\rn{{Q}}{P^i}\right)\right]\bigg)+
\\
&-\bigg(\lambda_{X^i}  E_{Q^i_{X^i}}[X^i]+ E_{P^i}\left[\Phi^i\left(\lambda_{X^i}\rn{Q^i_{X^i}}{P^i}\right)\right]\bigg)\bigg \}\leq \lambda_{X^i} \Big(  E_{Q^i_{X^i}}[Y^i]\Big).
\end{align*}
The concavity of $F^i_{Y^i}$ implies that, for all $i$,
$$U^i(X^i;Y^i)-U^i(X^i;0)=F^i_{Y^i}(1)-F^i_{Y^i}(0) \leq \frac{\mathrm{d}^+}{\mathrm{d}\alpha} F^i_{Y^i} \Big |_{\alpha=0} (1-0) \leq \lambda_{X^i} \Big(  E_{Q^i_{X^i}}[Y^i]\Big).$$
Thus from the condition in item 2. we get $E_{Q^i_{X^i}}[Y^i] \geq 0 $ for all $i$ and $E_{Q^j_{X^j}}[Y^j]>0$. Thus $\sum_{i=1}^N E_{Q^i_{X^i}}[Y^i]>0$.

To show 3.  $\Rightarrow $ 1., suppose that  
there exists 
$Y \in \mcY$ such that $\sum_{j=1}^N E_{Q^j_{X^j}}[Y^j] >0$. 
Define $\widehat{Y} $ by $\widehat{Y}^i:=Y^i+ \frac{1}{N} \sum_{j=1}^N E_{Q^j_{X^j}}[Y^j] -E_{Q^i_{X^i}}[Y^i] $, for all $i=1,\dots, N$, and note that $\widehat{Y} \in \mcY$, since $\R^N_0 \subseteq \mcY$ and $\mcY$ is a convex cone.
 Define the function $F_{\widehat{Y}^i} ^i(\alpha)=U^i(X^i;\alpha \widehat{Y}^i)$ and note that, by Proposition \ref{corgatbis},
it is differentiable at zero with
 $$ \frac{\mathrm{d}}{\mathrm{d}\alpha} F_{\widehat{Y}^i}^i \Big |_{\alpha=0} = \lambda_{X^i} E_{Q^i_{X^i}}[\widehat{Y}^i] =\lambda_{X^i} \Big( \frac{1}{N} \sum_{j=1}^N E_{Q^i_{X^j}}[Y^j] \Big) >0,$$
 as $\lambda_{X^i}>0$.
This implies that for each $i$ there exists some $\alpha^i >0$ such that $F_{\widehat{Y}^i}^i(\alpha)>F_{\widehat{Y}^i}^i(0)$ for all $\alpha \in (0,\alpha^i]$. By taking $\widehat{\alpha}:=\min(\alpha^1,\dots,\alpha^N)>0$ we deduce that for all $i$ we have $F_{\widehat{Y}^i}^i(\widehat{\alpha})>F_{\widehat{Y}^i}^i(0)$, namely  $U^i(X^i;\widehat{\alpha} \widehat{Y}^i)>U^i(X^i;0)$ for all $i$, with $\widehat{\alpha} \widehat{Y} \in \mcY$.
 \end{proof}

\section{The utility maximization with random endowment of a single agent}
\label{sction:theutilmax}

In this Section we are considering the utility maximization of agent $i\in\{1,\dots,N\}$.
Under the Standing Assumptions, we shall construct an ambient Banach lattice $L_i$ and a convex cone
$\mathcal K_i$ so that Assumption \ref{assOK} is fulfilled and we may apply the results in Section \ref{secharbenefits}.
\textbf{In the remainder of the paper, the following assumption is understood to hold for each $i$ without further mention.}
 \begin{assumption}\label{assNew}
     We let $\mathcal T \subseteq [0,T]$ be the set of trading times, for a finite horizon $T>0$. For each $i$, $(\Omega,\mathcal{F}^i,\mathbb F^i=(\mathcal{F}^i_{t})_{t\in \mathcal T},\Pp)$ is a filtered probability space, with $\mathcal{F}^i=\mathcal{F}^i_T$ and $\mathcal F^i_0$ the trivial sigma-algebra, and $S^i$ is an $\mathbb{R}^{d_i}$-valued c\`adl\`ag semimartingale adapted to $\mathbb F^i$ and $d_i \in \mathbb N$.
     In continuous time we assume in addition that $\mathbb F^i $ satisfies the usual hypotheses.

 \end{assumption}

In discrete time, by the usual embedding of $\{0,1,\dots,T\} $ in $[0,T]$ , see for example \cite{DS2006} Section 7.2, we may also apply the result of this section to the discrete time setting in Section \ref{secsettingdiscrete}.\\

The utility maximization problem in the general semimartingale market with random endowment is to compute
\begin{equation}\label{primal2}
U^i_{W^i}(X^i;0)
:=\sup_{H\in \mathcal{H}^{W^i}_i}
\mathbb{E}_{P^i}\left[
u^i\left( X^i+\int_{0}^{T}H_t\,dS^i_t\right)
\right],
\end{equation}
where the portfolio process $H$ is $\mathbb{R}^{d_i}$-valued and belongs to a suitable class $\mathcal{H}^{W^i}_i$
of admissible integrands, defined below, where the random variable $W^i$ controls the trading losses.
Finally, $X^i$ is an $\mathcal{F}^i$-measurable random variable representing the random endowment.
In order for \eqref{primal2} to be well posed and to admit a dual representation, it is necessary to identify
a suitable ambient space $L_i$ in which the random endowment (and, later on, the relevant exchange variables)
are required to live, together with an appropriate choice of the admissible class $\mathcal{H}^{W^i}_i$.

We now briefly recall a well-established approach to this problem, developed for instance in
\cite{BF05,BFG11,BF08}, where all details can be found.

\textbf{Admissible strategies}\\
Given a fixed non--negative random variable $W^i\in L^0(\Omega,\mathcal{F}^i,\probp^i)$, the
domain of optimization for the primal problem \eqref{primal2} is the class of
$W^i$--\emph{admissible strategies}, defined as
\begin{equation}
\mathcal{H}^{W^i}_i:=\left\{ H\in L(S^i)\mid \exists \alpha>0\text{\mbox{ such that }}%
\int_{0}^{t}H_{s}\mathrm{d}S^i_{s}\geq - \alpha W^i, \, \forall t\in \mathcal T \right\} ,
\label{HW}
\end{equation}%
where $L(S^i)$ denotes the class of $\mathbb{F}^i$-predictable, $S^i$-integrable processes. In
other words, the random variable $W^i$ acts as a loss bound, uniform in time, for the trading
gains process $\int H\,\mathrm{d}S^i$.

It is immediate that, if $W^i$ is bounded, the class $\mathcal H^{W^i}_i$ coincides with $\mathcal H^{1}_i$. 
Indeed, in this case the loss bound induced by $W^i$ is equivalent to a constant bound, so that one may take $W^i=1$ without loss of generality.

The next step in applying convex duality methods to \eqref{primal2} is to
reformulate the problem as an optimization over random variables, rather than
over stochastic processes. To this end, we introduce the convex cone of
terminal gains attainable by $W^i$--admissible strategies
\begin{equation}
K^{W^i}_i=\left\{ \int_0^T H_t \mathrm{d} S^i_t\mid H\in \mathcal{H}^{W^i}_i\right\}.
\label{k_def}
\end{equation}%
Thus $K_i^{W^i}\subseteq L^0(\Omega,\mathcal{F}^i,\probp^i)$ is a convex cone such that,
for every $k\in K_i^{W^i}$, there exists a constant $\alpha\geq 0$ satisfying
$k\geq -\alpha W^i$ a.s., as in item~3 of Assumption~\ref{assOK}.

To obtain a meaningful utility maximization problem, the loss control
variable $W^i$ is required to satisfy two additional conditions. The first one
depends only on the price process $S^i$ and ensures that the class of
$W^i$--admissible strategies is sufficiently rich for trading purposes:

\begin{definition}[\cite{BF05}, Definition 2]
\label{viability} A random variable $W^i\in L^0(\Omega,\mathcal{F}^i,\probp^i)$, $W^i\geq 1$ a.s., is \emph{suitable} (for the
process $S^i$) if there exists a process $H\in
L(S^i)$ such that 
\begin{equation}
P^i(\{\omega \mid \exists t\geq 0\ \text{\mbox{such that }}H_{t}(\omega
)=0\})=0  \label{suit1}
\end{equation}%
and 
\begin{equation}
\left\vert \int_{0}^{t}H_s\mathrm{d}S^i_s \right\vert \leq W^i,\quad \forall t\in
\mathcal T.  \label{suit2}
\end{equation}
\end{definition}

The second requirement involves only the utility function and quantifies the
extent to which the investor is willing to tolerate large trading losses.

\begin{definition}[\cite{BF05}, Definition 3]
\label{compatibility}A positive random variable $W^i\in L^0(\Omega,\mathcal{F}^i,\probp^i)$ is \emph{\ strongly
compatible} with the utility function $u$ if 
\begin{equation}
E_{P^i}[u^i(-\alpha W^i)]>-\infty \mbox{ for all }\alpha >0.  \label{compatible}
\end{equation}

\end{definition}

For $\widehat{u^i}(x):=u^i(0)-u^i(-\abs{x})$, the space
$M^{\widehat{u}^i}(\Omega,\mathcal{F}^i,\probp^i)$ denotes the Orlicz Heart
associated with the Young function $\widehat{u^i}$ (see Appendix \ref{approlicz}). Thus the condition of being strongly compatible can be equivalently formulated as: 
$$W^i \in M^{\widehat{u}^i}(\Omega,\mathcal{F}^i,\probp^i). $$
\begin{remark}\label{remW1}
 When $S^i$ is locally bounded, the choice $W^i=1$ is automatically suitable and
strongly compatible (see \cite[Proposition 1]{BF05}). In this case, one
recovers the classical notion of admissible trading strategies, namely those
whose wealth process is bounded from below by a constant.   
\end{remark}

\begin{remark}\label{remW}
The choice $W^i=1$ leads to take as the ambient space $L_i=L^{\infty}(\Omega,\mathcal{F}^i,\probp^i)$.
However, in the particular case of a discrete-time market, as soon as the jumps of the process $S^i$ are unbounded then $S^i$
fails to be locally bounded. Consequently, selecting $W^i=1$ (and hence
$L=L^{\infty}(\Omega,\mathcal{F}^i,\probp^i)$) would yield a trivial optimization problem, since in this case
$K^{W^i}_i=\{0\}$ for $W^i=1$. See for instance Examples 1, 3.1, 3.2 in \cite{BF05} or Examples 4, Sec.5 in \cite{BF08}.
\end{remark}

In the non locally bounded setting, the existence of a suitable and strongly compatible
loss variable is no longer automatic and is typically linked to integrability
conditions on the jumps of the process $S^i$. We refer to the cited
references for further discussion.

 \begin{assumption}\label{ass12}    
For each $i$, there exists $W^i \in M^{\widehat{u}^i}(\Omega,\mathcal F^i,\probp^i)$  suitable
with $S^i$ and we adopt the
corresponding class of $W^i$-admissible strategies associated with this loss bound.

 \end{assumption}
 
For the loss bound $W^i$ satisfying Assumption \ref{ass12}, we shall distinguish two cases:
\begin{stupa}
    $W^i \in L^{\infty} (\Omega,\mathcal{F}^i,\probp^i)$. In this case w.l.o.g. we select $W^i=1$ and choose $L_i=L^{\infty}(\Omega,\mathcal F^i,P^i)$ and  $L^*_i=L^{1}(\Omega,\mathcal F^i,P^i)$ and adopt the dual pair $(L_i,L^*_i)$.
\end{stupa}
\begin{stupb}
    $W^i \notin L^{\infty}(\Omega,\mathcal F^i,P^i)$. In this case we select $L_i=M^{\widehat{u}^i}(\Omega,\mathcal F^i,P^i)$ and  $L^*_i=L^{\widehat{\Phi}}(\Omega,\mathcal F^i,P^i)$ and adopt the dual pair $(L_i,L^*_i)$.
\end{stupb}

In both setups and by taking
 $$\mathcal K_i:=K^{W^i}_i,$$ 
 as the convex cone of admissible payoffs,
 we obtain that \textbf{Assumption \ref{assOK} is satisfied}
and we may apply the results of Section \ref{secharbenefits}.\\

As the cone $\mathcal K_i=K^{W^i}_i$ is not in general contained in $L_i$, we introduce the convex cone $\mathcal C_i \subseteq L_i$ defined as
\begin{equation}\label{CiOK} 
\mathcal C_i:=\mathcal C^{W^i}_i:=((K^{W^i}_i-L_{+}^{0}(\Omega,\mathcal{F}^i,\probp^i))\cap L_i.
\end{equation}
 Under Assumption \ref{ass12}, the Lemma \ref{l2} guarantees that for every $X^i\in L_i$
\begin{align*}
U^i_{W^i}(X^i;0):=\sup_{k\in K^{W^i}_i=\mathcal K_i}E_{P^i}\left[ u^i\left( X^i+k\right) \right]
=\sup_{f\in \mathcal C^{W^i}_i=\mathcal C_i}E_{P^i}\left[ u^i\left( X^i+f\right) \right].
\end{align*}
We have thus rewritten our original problem \eqref{primal2} in a way consistent with our set up in Section \ref{secmainres}. Recalling that $\Phi^i$ is the convex conjugate of $u^i$, from Proposition \ref{corgatbis} we finally obtain that under Assumption \ref{ass12} 

\begin{align}
U^i_{W^i}(X^i;0)=&\sup_{f\in\mathcal{C}_i}E_{P^i}[u^i(X^i+f)]
= \min_{\lambda >0} \min_{\substack{  \probq\ll\probp^i,\\ \rn{Q}{P^i} \in \mathcal{C}_i^{0}}}
\Bigg (\lambda E_{Q}[X^i]+E_P\bigg[\Phi^i\Big(\lambda \rn{\probq}{\probp^i}\Big)\bigg] \Bigg )\label{dual13}
 \\ 
=&
\lambda_{X^i}E_{Q^i_{X^i}}[X^i]+E_P\bigg[\Phi^i\Big(\lambda_{X^i} \rn{\probq_{X^i}}{\probp^i}\Big)\bigg],\label{dual14}
\end{align}
for a unique pair $(\lambda_{X^i},\probq^i_{X^i})$ of minimizers, with $\lambda_{X^i}>0$,   $\probq^i_{X^i} \ll\probp^i$ and such that $\rn{Q^i_{X^i}}{P^i} \in \mathcal{C}_i^{0}$, the polar cone of $\mathcal C_i$ in the dual pair $(L_i,L^*_i)$. Moreover, the value $U_{W^i}
(X^i;0)$ does not depend on which $W^i$ satisfying Assumption \ref{ass12} is selected, see \cite{BF05} Theorem 1, item b). 

 It remains to specify the properties of the dual elements $\probq^i_{X^i}\ll\probp^i$ such that $\rn{Q^i_{X^i}}{P^i} \in \mathcal{C}_i^{0}$. 
Recall that an $\R^d$-valued semimartingale $S^i$ is called a sigma-martingale if there exists a predictable process $\phi^i$ taking values in $(0,+\infty)$ such that the $\R^d$-valued stochastic integral $(\phi \cdot S^i)$ is a martingale. 
Let us define\\
\begin{align*}
    M_{\sigma}(S^i)&:=\{Q \ll P^i \mid S^i \text{ is a } \mathbb{F}^i\text{-sigma }  \text{martingale under Q} \},\\
    M_{e,\sigma}(S^i)&:=\{Q \sim P^i \mid Q \in M_{\sigma}(S^i)\},\\
    \mathcal P_{\Phi^i}(P^i)&:=\Big \{Q \ll P^i \mid  E_{P^i}\left[ \Phi ^i\left( \lambda \rn{Q}{P^i}\right) \right]<+ \infty \text{ for some } \lambda >0 \Big \}.
\end{align*}

We also recall that when $S^i$ is bounded (resp. locally bounded) then $M_{\sigma}(S^i)$ coincides with the set $M(S^i)$ (resp. $M_{loc}(S^i)$) of martingale (resp. local martingale) measures.\\

\begin{remark}
The set $M_{\sigma}(S^i)$ also depends on the underlying filtration $\mathbb{F}^i$, although this dependence is not made explicit in the notation. In particular, even if two agents have access to the same price processes, say $S^i = S^j$, the corresponding sets $M_{\sigma}(S^i)$ and $M_{\sigma}(S^j)$ may differ whenever their information structures are different, i.e., $\mathbb{F}^i \neq \mathbb{F}^j$.
\end{remark}

Let $W^i$ satisfy Assumption \ref{ass12}. As shown in \cite{BF08}, Remark 18 and Proposition 19 (d), in both setup \textbf{A} and \textbf{B}, the normalized polar 
of $\mathcal C_i$, for $\mathcal C_i=C^{W^i}_i$, is contained in $(L^*_i)_+$ and is given by 
\begin{eqnarray*}
    \Big \{\rn{Q}{P^i} \in L^{*}_i \mid E_Q[f] \leq 0 \, \forall f \in \mathcal C_i \Big \}
    &=&\Big \{\rn{Q}{P^i} \in L^{*}_i \mid E_Q[k] \leq 0 \, \forall k \in \mathcal K_i \Big \}\\
    &=&\Big \{\rn{Q}{P} \in L^{*} \mid Q \in M_{\sigma}(S^i) \Big \}.
\end{eqnarray*}
Thus the unique minimizers $\lambda _{X^i}$ and $Q^i_{X^i} $ in \eqref{dual14}  satisfy $\lambda _{X^i}>0$ and $Q^i_{X^i} \in {M}_{\sigma}(S^i) $ with $\rn{Q^i_{X^i}}{P^i} \in L^{*}_i$.

We summarize in the following proposition what we have just proved.  Notice that the last assertion follows from a standard argument: see, for instance, Theorem 2 \cite{F00} and Remark 2 in \cite{BF05}).

\begin{proposition}\label{propMs} Under assumption \ref{ass12}, we consider both setup \textbf{A} or \textbf{B} and take $\mathcal K_i:=K^{W^i}_i$. Then Assumption \ref{assOK} and \eqref{dual13} and \eqref{dual14} hold true and the unique minimizer $Q^i_{X^i}$ is a sigma martingale measure for $S^i$, namely $Q^i_{X^i} \in {M}_{\sigma}(S^i) $, and satisfies $\rn{Q^i_{X^i}}{P^i} \in L^{*}_i$.
If additionally $M_{e,\sigma} (S^i)\cap \mathcal P_{\Phi^i}(P^i) \neq \emptyset$, then $Q^i_{X^i} \in {M}_{e,\sigma}(S^i) $, namely the optimizer $Q^i_{X^i} $ is additionally equivalent to $P^i$.
\end{proposition}

\section{Benefits from Collective Free Lunches or Arbitrages}
\label{Benefits}
\label{BenefitsfromCFL1}

The main purpose of this section is to interpret Theorem~\ref{corben} in a setting where multiple agents, in addition to their individual investment opportunities, may cooperate through risk sharing or mutually beneficial exchanges.\\
Generally speaking, arbitrage opportunities are expected to generate beneficial exchanges. In the collective framework, such opportunities are formalized through the notions of Collective Arbitrage $\mathbf{CA}(\mcY)$ and Collective Free Lunch $\mathbf{CFL}(\mcY)$. 
For the financial motivation behind the notion of Collective Arbitrage, as well as for the related concept of Collective Free Lunch, we refer the reader to the cited references, where these notions are introduced and their economic significance and relevance is discussed in detail.

The fact that the conditions of Collective Free Lunch $\mathbf{CFL}(\mcY)$ and Collective Arbitrage $\mathbf{CA}(\mcY)$ —which are weaker than the corresponding notions of individual Free Lunch and Arbitrage—imply the existence of beneficial exchanges is formalized in Propositions  \ref{NAben}, \ref{propIncreasingTerdiscr}.

A noteworthy feature of the theory is that beneficial exchanges may still arise even in the absence of such opportunities, i.e., under $\mathbf{NCA}(\mcY)$ or $\mathbf{NCFL}(\mcY)$. This phenomenon is illustrated in Corollaries~\ref{corbenbis} and~\ref{corbentris}, as well as in the examples of Section \ref{secexamplesdiscr}

The framework introduced below will be used in the two following subsections, where the analysis is carried out in discrete and continuous time, respectively.

In the continuous-time setting, we assume that the price processes $S^i$ are locally bounded. Under this assumption, it is natural to work with local martingale measures, since
\begin{align*}
M_{\sigma}(S^i)
:&= \{Q \ll P^i \mid S^i \text{ is an } \mathbb{F}^i\text{-sigma martingale under } Q \} \\
&= \{Q \ll P^i \mid S^i \text{ is an } \mathbb{F}^i\text{-local martingale under } Q \}
=: M_{\mathrm{loc}}(S^i).
\end{align*}

In discrete time, no additional assumptions on $S^i$ are required: it is well known (see Proposition~\ref{propMsigma} Item 1) that $S^i$ is a $\sigma$-martingale if and only if it is a local martingale. Consequently, in this case as well,
\[
M_{\sigma}(S^i) = M_{\mathrm{loc}}(S^i).
\]

Therefore, throughout this section we shall work with the set $M_{\mathrm{loc}}(S^i)$ in place of $M_{\sigma}(S^i)$ introduced in Section~2.\\

\noindent \textbf{The segmented financial Market}

We continue to work under Assumption \ref{assNew}.
We postulate that each agent $i$, $i=1,\dots,N$, has preferences represented by expected utility assigned by $(u^i,P^i)$ and satisfying the Standing Assumptions and we denote by $\fib=(\fit)_{t\in\mathcal{T}}\subseteq \mathbb F$ the filtration representing the information available to agent $i$ with  $\mathcal{F}^i=\mathcal{F}^i_T$ and $\mathcal F^i_0$ the trivial sigma-algebra. 
We consider a global financial market consisting of one riskless asset and $J$ risky assets. We assume that each agent has access only to a given subset of the available risky assets. In particular, we do not exclude that different agents may invest in common assets, nor that some agents (or all agents) may have access to the entire set of traded assets.
For each agent $i$, let $d_i$ denote the number of risky assets available to that agent, and assume that
\[
1 \leq d_i \leq J \quad \text{for all } i.
\]
To model this segmented market structure, let $S^0_t = 1$ for all $t \in \mathcal T$ denote the riskless asset, and let
\[
S^i = (S^i_t)_{t \in \mathcal T}
\]
be the (discounted) price dynamics of the $d_i$ risky assets accessible to agent $i$ (and satisfying Assumption \ref{assNew}).
When needed, we will denote with $\mathbf S$ the $J$ risky assets available in the global financial market.\\
For a detailed and exhaustive motivation behind the segmented market setup the reader may consult \cite{Carassus23}, \cite{DFM25} or \cite{BDFFM25}.

\begin{remark}
We assumed that the processes $S^i$ are adapted to the filtration $\mathbb F^{\,i}$. 
This baseline requirement is flexible enough to accommodate additional structural relations among the
filtrations that may arise in specific applications. For instance, one may impose (for some or all pairs
$i,j$) that $\mathbb F^{\,i}=\mathbb F^{\,j}$, or only that the terminal $\sigma$--algebras coincide,
$\mathcal F^{\,i}_T=\mathcal F^{\,j}_T$. One may also posit the existence of a common subfiltration
$\mathbb G$ shared by all agents, e.g.\ $\mathbb G\subseteq \mathbb F^{\,i}$ for every $i$, where $\mathbb G$
could represent publicly available information. Alternatively, one could assume the existence of a common terminal-time sub $\sigma$-algebra $\mathcal G_T$ shared by all agents, e.g.\ $\mathcal G_T\subseteq \mathcal F_T^{\,i}$ for every $i$, generated for example by the exchange family
$\mathcal Y$.
Regardless of which additional relations between the filtrations are assumed, the only standing
mathematical conditions needed for the results below are:
$S^i$ is $\mathbb F^{\,i}$--adapted for each $i$ and $\mathbb F^{\,i}\subseteq \mathbb F$.
\end{remark}

\begin{notation}
\label{productnotation}
Let $\mathbf P=(P^1,\dots,P^N) \in \mathcal (P_{e})^N $ and $p \in [0,\infty]$. We denote with $L^{bb}(\Omega,\mathcal F^i,P^i)$ the cone of random variables in $L^{0}(\Omega,\mathcal F^i,P^i)$ that are a.s. bounded from below by a constant.
We set 
$$L^{p }(\Omega, \mathbf{F}_t,\mathbf P):=L^{p }(\Omega, \mathcal{F}_t^{1},P^1) \times \dots \times  L^{p }(\Omega, \mathcal{F}_t^{N},P^N), \quad t\in \mathcal T .$$ 
When we consider a single probability measure $Q \sim P$ on $\mathcal F$, we use  $$L^{p }(\Omega, \mathbf{F}_t,Q):=L^{p }(\Omega, \mathcal{F}_t^{1},Q) \times \dots \times  L^{p }(\Omega, \mathcal{F}_t^{N},Q),\quad t\in \mathcal T .$$
Similarly for $L^{bb}(\Omega,\mathbf F_t,Q)$ and $L^{bb}(\Omega,\mathbf F_t,\mathbf P)$.
\end{notation}

\subsection{The discrete time framework}
\label{secsettingdiscrete}


In addition to the segmented market described in Section \ref{Benefits} in this section we assume that $$\mathcal T=\{0, 1,\dots, T\}$$ is the finite set of trading times.  
In such discrete time setting the discounted price process $ S^i = (S^i_t)_{t \in \{0, 1,\dots, T\}}$
is an $\mathbb{R}^{d_i}$-valued $\mathbb{F}^i$-adapted stochastic process on
$
(\Omega, \mathcal{F}^i, \mathbb{F}^i := (\mathcal{F}^i_t)_{t \in \{0, 1,\dots, T\} }, \Pp)$, with $\mathcal{F}^i=\mathcal{F}^i_T$ and $\mathcal F^i_0$ the trivial sigma-algebra.

The class $ \mathbb H^i$ of admissible trading strategies for the agent $i$ is simply the class of $d_i$ - dimensional $\mathbb F^i$-predictable processes: 
\begin{equation}\label{HH}
\mathbb H^i:=\{H \text{ is } d_i\text{-dimensional } 
\mathbb F^i\text{-predictable process on }  (\Omega, \mathcal{F}^i, \mathbb{F}^i := (\mathcal{F}^i_t)_{t \in \{ 1,\dots, T\} }, \Pp) \}.    
\end{equation}
The stochastic integral with respect to such integrand $ H \in \mathbb H^i$ becomes the finite sum 
\begin{equation*}\int_{0}^{t}H_{s}\mathrm{d}S^i_{s}=(H \cdot S^i)_t:=
\sum_{s=1 }^t H_s(S^i_s-S^i_{s-1}), \, \, t \in \{1,\dots,T\},  
\end{equation*} 
 and consequently the space of terminal payoffs with zero initial cost is
\begin{equation}
\label{def:KiBis}
\mathbb K_i:=\{(H\cdot S^i)_T \mid H \in {\mathbb H^i} \} \subseteq L^{0 }(\Omega, \mathcal{F}^{i},P^i).
\end{equation}
The set of martingale and local martingale measures for  $S^i$ are defined by
\begin{equation*}
M(S^i):=\left\{ Q \in \mathcal{P}_{ac} \mid     S^i  \text { is a }  \fib \text {-martingale under } Q \right\},
\end{equation*}
\begin{equation*}
M_{\mathrm{loc}}(S^i):=\left\{ Q \in \mathcal{P}_{ac} \mid     S^i  \text { is a }  \fib \text {-local martingale under } Q \right\},
\end{equation*}
and we set $M_e(S^i):=\mathcal P_e \cap M(S^i) $ and $M_{e,\mathrm{loc}}(S^i):=\mathcal P_e \cap M_{\mathrm{loc}}(S^i)$.
The classical No Arbitrage condition for agent $i$ holds if:
\begin{equation*}
\label{eq:NAi}
 \mathbf{NA}_{i}: \text{  } \mathbb K_i \cap L_{+}^{0}(\Omega, \mathcal{F}^i , P^i)=\{0\}.
 \end{equation*}
 The $\mathbf{NA}_{i}$ condition implies that the set $(\mathbb K_i - L_{+}^{0}(\Omega, \mathcal{F}^i , P^i))$ is closed in $L^{0}(\Omega, \mathcal{F}^i , P^i)$ (see \cite{DS2006} Theorem 6.9.2), an essential property for the proof of both the first FTAP and the super-hedging duality.
\begin{theorem}[Dalang, Morton and Willinger (1990) \cite{DMW90}]\label{DMW}
In the discrete time setting described above, for any $i$ and if $S^i$ is integrable under $P^i$, it holds 
\begin{equation*}
 \mathbf{NA}_{i} \;\Longleftrightarrow\;  M_e(S^i) \not = \emptyset .
\end{equation*}
\end{theorem}
It is well-known that the integrability of $S^i$
  can be ensured by a change to an equivalent probability measure.
Similar definition and results holds for the global market $\mathbf S$ and we let $\mathbb K$ be the set of time $T$ stochastic integrals, with $\mathbb F$-predictable integrands, with  respect to all the assets  $\mathbf S$ and let $M(\mathbf S) $ be the set of martingale measures for the whole market $\mathbf S$ with respect to the  filtration $\mathbb F$.
We say that there is No Global Arbitrage, denoted with $\mathbf{NA}$, if there is no classical arbitrage in the whole market $\mathbf S$:
\begin{equation*}
 \mathbf{NA}: \mathbb K \cap L_{+}^{0}(\Omega, \mathcal{F} , \Pp)=\{0\}.
 \end{equation*}
When agent $i$ follows an investment strategy in its own market $S^i$, they will obtain a terminal payoff $k^i \in {\mathbb K}_i$.
The agents may also enter in the risk exchange corresponding to a vector $Y \in \mathcal Y $. This procedure leads to the terminal time value $k^i + Y^i$ for agent $i$.
In \cite{BDFFM25} a \emph{Collective Arbitrage} consists of vectors $(k^1, \dots , k^N) \in {\mathbb K}_1 \times \cdots \times {\mathbb K}_N$ and $Y =(Y^1, \dots, Y^N) \in \mcY$ satisfying
\begin{align*}
 k^i+Y^i\geq 0  \quad \text{a.s. } & \quad \forall \, i\in{1,\dots,N} \, \text{ and } 
\\ P^n(k^n+Y^n>0)>0   &\quad \text{for at least one } n \in {1,\dots,N}.
\end{align*}
Thus the absence of collective arbitrage is formulated as follows.
 \begin{definition}[No Collective Arbitrage, Def 3.1 \cite{BDFFM25}]
\label{NCA} 
No Collective Arbitrage for $\mathcal{Y}$ ($\mathbf{NCA}(\mathcal{Y})$) holds if
\begin{equation*}
({\sf X}_{i=1} ^{N} \mathbb K_i+\mathcal Y )\cap  L^{0 }_+(\Omega, \mathbf{F},\mathbf P)=\{0\}, \label{NCAY} 
 \end{equation*}
 where ${\sf X}_{i=1} ^{N} \mathbb K_i$ denotes the Cartesian product of the sets $\mathbb K_i$ defined in \eqref{def:KiBis}.
 \end{definition}

 As shown in Proposition 3.2  \cite{BDFFM25}, 
 \begin{equation}\label{3equi}
\mathbf{NCA}(\mathcal{Y}) \Leftrightarrow \mathbb K^\mathcal Y \cap  L^{0 }_+(\Omega, \mathbf{F} , \mathbf P)=\{0\} \Leftrightarrow \mathbb C^\mathcal Y \cap  L^{1 }_+(\Omega, \mathbf{F} , \mathbf P)=\{0\} ,  
\end{equation}
 where 
 \begin{equation}\label{KK}
      \mathbb K^\mathcal Y:=  {\sf X}_{i=1} ^{N} (\mathbb  K_i - L^{0 }_+(\Omega, \mathcal{F}^i,P^i) ) + \mathcal Y, \quad  \mathbb C^\mcY:=K^\mcY\cap L^1(\Omega,\mathbf{F},\mathbf P).
 \end{equation} 

The implications
\begin{equation}\label{Implications}
\mathbf{NA} \Rightarrow \mathbf{NCA(\mathcal{Y})} \Rightarrow \mathbf{NA}_{i} \; \forall\, i\in{1,\dots,N},
\end{equation}
essentially always hold: the former only requires $\mcY\subseteq \mcY_0$, the latter holds under the even weaker assumption that $0\in\mcY$. Moreover, it was shown in the aforementioned reference that
\begin{align}
     \mathbf{NA} & \Longleftrightarrow  \mathbf{NCA}(\mcY) \,\text{ if }\, \mcY=\mcY_0 \text{ and }\mathbb{F}^i=\mathbb{F}^j=\mathbb{F}, \text{ for all } i,j. \label{12}
   \\  \mathbf{NA}_{i} \text{ } \forall \, i & \Longleftrightarrow   \mathbf{NCA(\mcY)} \,\text{ if }\, \mcY=\mathbb R_0 ^N.\label{NAiNCA}
\end{align}
 However, the notions of $\mathbf{NCA}(\mathcal{Y})$ generate novel concepts for general sets $\mcY$.

\begin{remark}[Collective arbitrage vs.\ representative-agent arbitrage] 
A natural conjecture is that \emph{collective no–arbitrage} for many agents is equivalent to \emph{classical
no–arbitrage} for a suitably defined representative (social–planner) agent. Intuitively, one aggregates all
individual trading opportunities and internal zero–sum exchanges into a single “social” budget constraint and
asks whether any zero–cost position yields an a.s.\ nonnegative terminal payoff with strictly positive payoff on
a set of positive probability.\\
Under the structural conditions collected in~\eqref{12} (e.g., common 
information structure and a shared comprehensive cone of allowable exchanges), the existence of a \emph{common} separating price system that
prices \emph{all} agents’ admissible payoffs and exchanges is possible and $\mathbf{NCA}(\mcY)$ is equivalent to the absence of arbitrage for the representative agent; see
\cite{BDFFM25}, Section 6.2, for a precise statement and proof. Absent these alignment conditions, the two notions
need not coincide,
as the aggregation step is not innocuous.
A particularly transparent case is~\eqref{NAiNCA}, where the ``no–collective–arbitrage'' condition
$\mathbf{NCA}(\mathcal Y)$ is equivalent to $\mathbf{NA}_i$ \emph{for each agent $i$}. Here, the collective
feasible set reduces to the cartesian product of the individual feasible sets, and there is no additional
aggregate flexibility beyond what each agent can already do. In such a setting, identifying a representative
agent that exactly reproduces the joint (multi–agent) feasibility and pricing relations is, in general,
impossible. 
\end{remark}

From \cite{DFM25}, Definition 3, we recall the definition of the set  $\Me$ of \textit{equivalent collective martingale measures} and introduce the set $ \mathcal M_{e,\mathrm {loc}}$ of \textit{equivalent collective local martingale measures}. 

\begin{definition}

\begin{equation*}\label{MartingaleMeasures}
\begin{split}
    \Me=\bigg\{ Q=(Q^1,\dots,Q^N) \in {\sf X}_{i=1} ^{N} M_e(S^i)   \mid  \mcY \subseteq L^{1 }(\Omega, \mathbf{F},\mathbf Q)  \text { and  } \sum_{i=1}^N E_{Q^i}[Y^i]\leq 0 \text {  }\forall Y \in \mcY  \bigg\}.
\end{split}
\end{equation*}

\begin{equation}\label{MartingaleLocalMeasures}
\begin{split}
    \mathcal M_{e,\mathrm {loc}}=\bigg\{ Q=(Q^1,\dots,Q^N) \in {\sf X}_{i=1} ^{N} M_{e,\mathrm{loc}}(S^i)   \mid  \mcY \subseteq L^{1 }(\Omega, \mathbf{F},\mathbf Q)  \text { and  } \sum_{i=1}^N E_{Q^i}[Y^i]\leq 0 \text {  }\forall Y \in \mcY  \bigg\}.
\end{split}
\end{equation}

\end{definition}

\begin{theorem}[ 
\cite{DFM25} Theorems  8 and  12,]\label{FTAP:R}
    Let $\mcY $ be a finite dimensional  vector space such that  $\mathbb R^N_0 \subseteq \mcY\subseteq L^{\infty}(\Omega, \mathbf{F},\mathbf P)$. Then 
    \begin{equation*}
      \mathbf{NCA(\mathcal{Y})} \iff \Me \not = \emptyset 
    \end{equation*}
    and $\mathbf{NCA(\mathcal{Y})}$ implies that $\mathbb K^{\mcY} \text{ is closed in } L^{0}(\Omega, \mathbf{F},\Pp)$.
\end{theorem}

In the next theorem, we establish a strengthened version of Theorem~\ref{FTAP:R}, which will be used to prove Proposition \ref{NAben}.

\begin{theorem}\label{FTAP:RBis}
    Let $\mcY $ be a finite dimensional  vector space such that  $\mathbb R^N_0 \subseteq \mcY\subseteq L^{\infty}(\Omega, \mathbf{F},\mathbf P)$. Then 
    \begin{equation*}
      \mathbf{NCA(\mathcal{Y})} \iff \Me \not = \emptyset \iff \mathcal M_{e,\mathrm {loc}}(\mcY) \not = \emptyset .
    \end{equation*}
   
\end{theorem}

\begin{proof}
Given  Theorem~\ref{FTAP:R}, and since  $\Me \subseteq \mathcal M_{e,\mathrm {loc}}(\mcY)$, we only need to prove $\mathcal M_{e,\mathrm {loc}}(\mcY) \not = \emptyset  \Rightarrow \mathbf{NCA(\mathcal{Y})} $.  
Take $(Q^1, \dots , Q^ N) \in  M_{e,\mathrm {loc}}(\mcY)$ and let $(k+Y) \in  ({\sf X}_{i=1} ^{N}  \mathbb K_i   + \mathcal Y ) \cap  L^{0 }_+(\Omega,  \mathbf{F},\mathbf P)$. Then $k^i+Y^i \geq 0$ and thus $k^i \geq -Y^i  \in L^{\infty }(\Omega, \mathcal{F}^i , P^i) $ for all $i$. As $k^i=(H \cdot S^i)_T \geq -C$, for some constant $C>0 $ and $H \in \mathbb H^i$, and $S^i$ is a $\mathbb F^i$-local martingale under $Q^i$, Proposition item 6 assures that $(H \cdot S^i)$ is a $\mathbb F^i$-martingale under $Q^i$ and so  $E_{Q^i}[k^i]=0$. Therefore $E_{Q^i}[k^i+Y^i] = E_{Q^i}[Y^i]$, $\sum_{i=1}^N E_{Q^i}[k^i+Y^i] = \sum_{i=1}^N E_{Q^i}[Y^i]  \leq 0$, as $(Q^1, \dots , Q^ N) \in M_{e,\mathrm {loc}}(\mcY)$.
From $(k^i+Y^i) \geq 0 $ for all $i$, we also get that $\sum_{i=1}^N E_{Q^i}[k^i+Y^i]  \geq 0$, so that $\sum_{i=1}^N E_{Q^i}[k^i+Y^i] = 0$, which then implies $E_{Q^i}[k^i+Y^i] = 0$ for all $i$ and $k^i+Y^i=0$ for all $i$. Thus $\mathbf{NCA(\mathcal{Y})}$ holds true.


\end{proof}

Up to this point, this subsection has introduced the classical notion of No Arbitrage, together with the concept of No Collective Arbitrage as developed in \cite{BDFFM25}. In that framework, we considered the class $\mathbb{H}^i$ of admissible trading strategies defined in \eqref{HH}, along with the associated set $\mathbb{K}_i$ of attainable terminal payoffs.
However, this class is not well suited for utility maximization problems, and in particular it does not allow for the direct application of convex duality techniques.
For the analysis of utility maximization carried out below, we instead adopt the framework introduced in Section~3. In particular, we consider the class
\begin{equation}
  \mathcal{H}^i = \mathcal{H}_i^{W^i}= \left\{ H\in \mathbb H^i \mid \exists \alpha>0\text{\mbox{ such that }}%
(H \cdot S^i)_t \geq - \alpha W^i, \, \forall t\in \mathcal T \right\} , 
\end{equation} 
of trading strategies defined as in \eqref{HW}, together with the corresponding sets 
\begin{equation}
\mathcal{K}_i = \mathcal{K}_i^{W^i}=\left\{ (H \cdot S^i)_T \mid H\in \mathcal{H}_{i}^{W^i}\right\}, \quad  \mathcal{C}_i=(\mathcal K_i -L^0_+(\Omega, \mathcal F^i,P^i) \cap L_i \,.
\label{k_defiTris}
\end{equation} 
of terminal payoffs  defined as in \eqref{k_def} and \eqref{CiOK}. In this way we are back in the setting of Section \ref{sction:theutilmax} and we may apply the results therein.

This choice entails no inconsistency because the condition of No Collective Arbitrage can be equivalently formulated in terms of strategies belonging to $\mathbb H^i$ or to $\mathcal H^i$, as shown in the following Proposition \ref{propNCAKK}. Moreover, the formulation of No Collective Arbitrage in the form
\[
\mathcal C^{\mathcal Y} \cap \mathcal L_+ = \{0\}
\]
(see \eqref{NCAKK})  will serve as the basis for strengthening such condition to the notion of No Collective Free Lunch in Section \ref{seccontinuous}.

Before presenting the result, observe that if there exists an arbitrage opportunity for some agent $j$, that is, $k^j \in \mathbb K_j$ with $k^j \geq 0 $ and $\mathbb{P}(k^j > 0) > 0$, then a collective arbitrage arises trivially: it suffices to take $Y = 0$ and $k^i = 0$ for all $i \neq j$.
The more interesting situation is that of a \emph{genuine} collective arbitrage, namely when $\mathbf{CA}(\mathcal{Y})$ holds even though each individual market is arbitrage-free see \cite[Section~3.1]{BDFFM25}. By Theorem~\ref{FTAP:RBis}, this corresponds to the condition
$\mathcal{M}_{e,\mathrm{loc}}(S^i) \neq \emptyset \quad \text{for all } i$ and this is precisely the case addressed in the Propositions \ref{propNCAKK}, \ref{NAben} and \ref{propcomplete}.

\begin{proposition}\label{propNCAKK}
 Suppose that  $\mathbb R^N_0 \subseteq \mcY\subseteq L^{\infty}(\Omega, \mathbf{F},\mathbf P)$ and   $M_{e,\mathrm{loc}}(S^i) \neq \emptyset $ for each $i$. Set $\mathcal L:={\sf X}_{i=1} ^{N}  (L_i)$ and $\mathcal L_+:={\sf X}_{i=1} ^{N}  (L_i)_+$. Then for any choice of suitable $1\leq W^i\in M^{\widehat{u}^i}(\Omega,\mathcal F^i,P^i) $, $i=1,\dots,N$
\begin{align}\label{NCAK}
&\mathbf{NCA} 
(\mathcal{Y})  \iff ({\sf X}_{i=1} ^{N} \mathbb K_i+\mathcal Y )\cap  L^{0 }_+(\Omega, \mathbf{F},\mathbf P)=\{0\}, \\ 
& \iff \mathbb K^{\mathcal Y} \cap  L^{0 }_+(\Omega, \mathbf{F},\mathbf P)=\{0\} \iff  \mathbb C^\mathcal Y \cap  L^{1 }_+(\Omega, \mathbf{F} , \mathbf P)=\{0\} \label{NCAh} \\
&\iff \mathcal C^\mathcal Y \cap \mathcal L_+ =\{0\} \iff ({\sf X}_{i=1} ^{N} \mathcal K_i+\mathcal Y )\cap  L^{0 }_+(\Omega, \mathbf{F},\mathbf P)=\{0\}, \label{NCAKK}
\end{align}
where $\mathcal K_i$ is defined in \eqref{k_defiTris} and
 \begin{equation}\label{KKbis}
  \mathcal K^\mathcal Y:=  {\sf X}_{i=1} ^{N} (\mathcal  K_i - L^{0 }_+(\Omega, \mathcal{F}^i,P^i) ) + \mathcal Y, \quad 
        \mathcal C^\mcY:=\mathcal K^\mcY\cap \mathcal L , 
 \end{equation} 
namely the conditions $\mathbf{NCA}(\mathcal{Y}) $ with strategies $H^i \in \mathbb H^i$ and No Collective Arbitrage, as in \eqref{NCAKK}, with strategies $H^i \in \mathcal H^i$ are equivalent.
\end{proposition}
\begin{proof}
The three equivalences in \eqref{NCAK} and \eqref{NCAh}  were already mentioned in \eqref{3equi} and were proven in \cite{BDFFM25} Proposition 3.2.
Exactly with the same proof it is possible to prove the last equivalence in \eqref{NCAKK}. Moreover, since $\mathcal K_i \subseteq \mathbb K_i$ it is obvious that
\begin{equation}
({\sf X}_{i=1} ^{N} \mathbb K_i+\mathcal Y )\cap  L^{0 }_+(\Omega, \mathbf{F},\mathbf P)=\{0\} \Rightarrow ({\sf X}_{i=1} ^{N} \mathcal K_i+\mathcal Y )\cap  L^{0 }_+(\Omega, \mathbf{F},\mathbf P)=\{0\}.
\end{equation}
Therefore it remains to prove only 

\begin{equation*}
({\sf X}_{i=1} ^{N} \mathbb K_i+\mathcal Y )\cap  L^{0 }_+(\Omega, \mathbf{F},\mathbf P)=\{0\} \Leftarrow ({\sf X}_{i=1} ^{N} \mathcal K_i+\mathcal Y )\cap  L^{0 }_+(\Omega, \mathbf{F},\mathbf P)=\{0\}.
\end{equation*}
or equivalently,
\begin{equation*}
\exists f \neq 0 \text{ s.t. }  f \in ({\sf X}_{i=1} ^{N} \mathbb K_i+\mathcal Y )\cap  L^{0 }_+(\Omega, \mathbf{F},\mathbf P) \Rightarrow   f \in ({\sf X}_{i=1} ^{N} \mathcal K_i+\mathcal Y )\cap  L^{0 }_+(\Omega, \mathbf{F},\mathbf P).
\end{equation*}
Let  $0 \leq f=k+Y \in ({\sf X}_{i=1} ^{N} \mathbb K_i+\mathcal Y )$, with $\Pp (f^j>0)>0 $ for some $j$, and fix $W^1,\dots, W^N$ in the statement.
Since $Y\in\mathcal{Y}\subseteq  L^{\infty}(\Omega,\mathbf F,\mathbf P)$ we have $k \in  L^{bb}(\Omega,\mathbf F,\mathbf P)$. In particular for every component we have $k^i=(H\cdot S^i)_T\geq -\norm{Y^i}_\infty$, where $H\in \mathbb{H}^i$. We show that $H\in \mathcal{H}_i^{W^i}$. To see this pick $Q^i\in M_{e,\mathrm{loc}} (S^i) \neq \emptyset$, so that $S^i$ is a $\mathbb F^i$-local martingale under $Q^i$. By Proposition \ref{propMsigma} Item 6, the process $(H\cdot S^i)$ is a $\mathbb{F}^i$-martingale under $Q^i$, so that 
$$(H \cdot S)_t=E_Q[(H \cdot S)_T|\mathcal{F}_t]\geq E_Q[-Y^i|\mathcal{F}_t]\geq -\norm{Y^i}_\infty\quad \forall t\in\mathcal{T}.$$
This ensures that $H\in \mathcal{H}_i^{1}\subseteq \mathcal{H}_i^{W^i}$.
Thus  $f \in ({\sf X}_{i=1} ^{N} \mathcal K_i+\mathcal Y )\cap  L^{0 }_+(\Omega, \mathbf{F},\mathbf P)$.
\end{proof}

We now move to the main results of this Section. We start stating explicitly all needed technical requirements, in addition to the tacitly assumed Standing Assumptions and Assumption  \ref{assNew}


\begin{assumption}\label{assdiscrtime}

Let $\mathcal T=\{0,1,\dots,T\}$. We suppose that, for each $i$,
    \begin{enumerate}

        \item  $S^i$ is an $\mathbb{R}^{d_i}$-valued stochastic process
defined on
$ (\Omega,\mathcal{F}^i,\mathbb F^i:=(\mathcal{F}^i_{t})_{t\in \{0,1,\dots,T\}},\Pp)$

        \item We adopt setup \textbf{B}, that is \begin{enumerate}
         \item $W^i\in M^{\widehat{u}^i}(\Omega,\mathcal F^i,P^i)$ is suitable for $S^i$ in the sense of Definition \ref{viability}.
            \item $L_i=M^{\widehat{u}^i}(\Omega,\mathcal F^i,P^i)$ and  $L_i^*=L^{\widehat{\Phi}^i}(\Omega,\mathcal F^i,P^i)$.
        \end{enumerate}
        \item $\mcY$ is a finite dimensional vector space with $\mathbb R_0 ^N \subseteq \mcY \subseteq  L^{\infty}(\Omega,\mathbf F,\mathbf P)$.
    \end{enumerate}
\end{assumption}
With the selection \eqref{k_defiTris} and recalling Proposition \ref{propMs}, we conclude that:
\begin{center}
    \textbf{Assumption~\ref{assdiscrtime} $\Rightarrow$ Assumption~\ref{ass12} $\Rightarrow$ Assumption~\ref{assOK}}

\end{center}
and the results of Sections \ref{secmainres} and \ref{secharbenefits} can be applied.

 \begin{corollary}[of Theorem \ref{corben}] \label{corbentris}
Under Assumption \ref{assdiscrtime}, the following are equivalent
 \begin{enumerate}
     \item There exists a strictly beneficial $\widehat{Y} \in \mcY$.
     \item There exists a beneficial $\widehat{Y} \in \mcY$.
     \item $\mathbf Q_X \notin \mathcal M_{\mathrm{loc}}(\mcY)$, where $\mathbf Q_X:=(Q_{X^1}^1,\dots,Q_{X^N}^N)$ is the minimax vector and 
     \begin{align*}
\mathcal M_{\mathrm{loc}}(\mcY) &:=\bigg\{ \mathbf Q=(Q^1,\dots,Q^N)  \mid Q^i \in  {M}_{\mathrm{loc}}(S^i) \, \text{ and }   \, \sum_{i=1}^N E_{Q^i}[Y^i] \leq 0 \text {   }\forall Y \in \mcY \bigg\}.
\end{align*}

\end{enumerate}
 \end{corollary}

\begin{proof}
    By Proposition \ref{propMs} and Proposition \ref{propMsigma} Item 1, $Q^i_{X^i} \in {M}_{\sigma}(S^i)={M}_{\mathrm{loc}}(S^i) $ satisfies $\rn{Q^i_{X^i}}{P^i} \in L_i^*=L^{\widehat{\Phi}^i}(\Omega,\mathcal F^i,P^i) \subseteq L^{1}(\Omega,\mathcal F^i,P^i)$. The thesis now follows from Theorem \ref{corben} noticing that $\mathbf Q_X \notin \mathcal M_{\mathrm{loc}}(\mcY)$ if and only if $\sum_{i=1}^N E_{Q_{X^i}^i}[Y^i]> 0 \text {  for some  } Y \in \mcY.$ 
\end{proof}



The following proposition shows that a Collective Arbitrage necessarily gives rise to beneficial exchanges. The result is a direct consequence of Theorem~\ref{corben} and will also be derived by an independent direct argument in Section \ref{sec411}.

\begin{proposition}\label{NAben}
 Suppose that $M_{e,\mathrm{loc}} (S^i)\cap \mathcal P_{\Phi^i}(P^i) \neq \emptyset$ for all $i$. Under Assumption \ref{assdiscrtime}, the  existence of a  collective arbitrage  $\mathbf{CA}(\mathcal{Y})$ implies the existence of strict beneficial exchanges.
\end{proposition}

\begin{proof}
The assumption $M_{e,\mathrm{loc}} (S^i)\cap \mathcal P_{\Phi^i}(P^i) \neq \emptyset$ implies, by Proposition \ref{propMs} that the minimizers $Q^i_{X^i} \in {M}_{e,\mathrm{loc}}(S^i) $ and $\frac{dQ^i_{X^i}}{dP^i} \in L^{\widehat{\Phi}^i}(\Omega,\mathcal F^i,P^i) \subseteq L^1(\Omega,\mathcal F^i,P^i)$ for each $i$. If $\sum_{i=1}^N E_{Q^i_{X^i}}[Y^i] \leq 0 \text {   }\forall Y \in  \mcY \subseteq L^{\infty} (\Omega,\mathbf F,\mathbf P)$ then $(Q^1_{X^1}, \dots,Q^N_{X^N}) \in \mathcal M^*_{e,\mathrm{loc}} (\mcY) \subseteq \mathcal M_{e,\mathrm{loc}} (\mcY)$, but this is impossible since $\mathcal M_{e,\mathrm{loc}} (\mcY) =\emptyset $ due to to the existence of a $\mathbf{CA}(\mathcal{Y})$ and Theorem \ref{FTAP:RBis}. Thus $\sum_{i=1}^N E_{Q^i_{X^i}}[Y^i] > 0 \text {  for some  } Y \in \mcY $. Theorem \ref{corben} then implies the thesis.
\end{proof}

We conclude this discrete-time detour with the following result, which explains how completeness of single agents' markets avoids existence of beneficial exchanges.

\begin{proposition}\label{propcomplete}
   Suppose that $M_{e,\mathrm{loc}} (S^i)\cap \mathcal P_{\Phi^i}(P^i) \neq \emptyset$  for all $i$ and that $\mathbf{NCA}(\mathcal{Y})$ holds true. Under Assumption \ref{assdiscrtime} and if each individual market is classically complete, then no beneficial exchanges exist.
\end{proposition}
\begin{proof}
 $\mathbf{NCA}(\mathcal{Y})$ implies $\mathbf{NA}_i$ for all $i$. Thus by the classical First and Second FTAP the set $M_{e}(S^i) $ is not empty and reduced to a singleton, say  $M_{e}(S^i)=\{Q_*^i \} $ for all $i$. In discrete time, completeness of the market implies that the probability space is generated by a finite number of atoms and thus ${M}_{e,\mathrm{loc}}(S^i)={M}_{e}(S^i)=\{Q_*^i \} $. Moreover, $\mathbf{NCA}(\mathcal{Y})$ implies $\Me \neq \emptyset$, by Theorem \ref{FTAP:R}. Since $\Me \subseteq M_{e}(S^1)\times \dots \times M_{e}(S^N)$ we necessarily have $\Me=\{(Q_*^1,\dots,Q_*^N)\}$. By the assumption $M_{e}(S^i) \cap \mathcal P_{\Phi^i}(P^i) \neq \emptyset$, the minimizers $Q^i_{X^i} \in {M}_{e,\mathrm{loc}}(S^i)={M}_{e}(S^i)$ and $\frac{dQ^i_{X^i}}{dP^i} \in L_i^*$, thus $Q^i_{X^i}=\{Q_*^i \}$, which implies  $Q_X \in \mathcal M^*_{\mathrm{loc}}(\mcY)$. The Corollary  \ref{corbentris} then implies the thesis.
\end{proof}

\subsubsection{An informative proof of existence of beneficial exchanges if a CA exists}\label{sec411}
We now explain \emph{how} beneficial exchanges are produced exploiting the existence of arbitrage opportunities in a more explicit fashion. Under Assumption \ref{assdiscrtime} and from Proposition \ref{propMs} and \eqref{dualdual} we see that
 \begin{equation*}
        U^i(X^i;Z)=\min_{\lambda >0} \min_{\substack{  \probq\sim\probp^i,\\ Q\in M_{e,\mathrm{loc}}(S^i)}}
\Bigg (\lambda(x^i+E_Q[Z])+E_{P^i}\bigg[\Phi^i\Big(\lambda \rn{\probq}{\probp^i}\Big)\bigg] \Bigg ) ,\, \forall\,Z \in L_i\text{ s.t. } U^i(x^i;Z)<u^i(+\infty).      
\end{equation*}

\begin{proposition}
\label{propIncreasingTerdiscr}
 Suppose that $ M_{e,\mathrm{loc}} (S^i)\cap \mathcal P_{\Phi^i}(P^i) \neq \emptyset$ and that Assumption \ref{assdiscrtime} holds. Then \textbf{CA}$(\mcY)$ implies existence of beneficial exchanges for any choice of suitable $1\leq W^i\in M^{\widehat{u}^i}(\Omega,\mathcal F^i,P^i) $, $i=1,\dots,N$.
\end{proposition}
\begin{proof}
Let  $(k,Y)$ be a \textbf{CA}$(\mcY)$ and fix $W^1,\dots, W^N$ in the statement. We can assume that $(k,Y)\in {\sf X}_{i=1} ^{N} \mathcal K_i+\mathcal Y$ by Proposition \ref{propNCAKK}.
Recall the definition $\mathcal{C}_i$ in \eqref{Ci}.
Consider now $n>2\max_i\norm{Y^i}_\infty$ and $f^i=\min(k^i,n)+Y^i=k^i-(k^i-n)1_{\{k^i>n\}}+Y^i$. Clearly 
\begin{align*}
  \min(k^i,n)=k^i-(k^i-&n)1_{\{k^i>n\}}\in 
\big(\mathcal{K}_i -L^{0 }_+(\Omega, {F}^i,P^i)\big)\cap L^{\infty}(\Omega, {F}^i,P^i)\\
&\subseteq  \big(\mathcal{K}_i -L^{0 }_+(\Omega, {F}^i,P^i)\big)\cap L_i=\mathcal{C}_i.  
\end{align*}
We show that $f^i\geq 0$ a.s. for every $i$ and $P^j(f^j>0)>0$ for some $j$. Since $f^i=(k^i+Y^i)1_{k^i\leq n}+ (n+Y^i)1_{k^i> n}$ and $n>2\norm{Y^i}_\infty$ the first inequality is proved. As for the second, recall that since  $(k,Y)$ is a \textbf{CA}$(\mcY)$ we have $P^j(k^j+Y^j>0)>0$ for some $j$. Now, we observe that $f^j\geq (n+Y^i)1_{k^i> n}\geq \norm{Y^i}_\infty1_{k^i> n}$ so that if $P^j(k^j>n)>0$, trivially $P^j(f^j>0)>0$. If instead $k^j\leq n$ a.s., we see that $f^j=k^j+Y^j$ a.s., and in particular again $P^j(f^j>0)>0$.

Take now $A_n:=\{f^j\geq \frac1n\}$. For some $N$ big enough we have $P(A_N)>0$. For the given set $A_N$ pick $C$ from Lemma  \ref{LemmaincreaseBis}, which holds true verbatim in the discrete-time setup. 
The properties of $f$ ensure that
$$U^i(x^i;cf^i)\geq U^i(x^i;0)\,\forall i\text{ and } U^j(x^j;cf^j)\geq U^j\bigg(x^j;\frac{c}{N}1_{A_N}\bigg)>U^j(x^j;0)$$
using Lemma \ref{LemmaincreaseBis} for the last inequality.
Under Assumption \ref{assdiscrtime} we have (from Proposition \ref{propU*}) finiteness and norm continuity of $U^i(x^i;\cdot)$ on $L_i$.
In particular, we can pick $a\in \R^N_0$ with $a^j<0$ and $a^i>0$
 for every $i\neq j$ in such a way that 
$$U^i\big(x^i;c(f^i+a^i)\big)> U^i(x^i;0)\,\forall i.$$
Note  that $c(f+a)\in \mathcal{C}_1\times\dots\mathcal{C}_N+\mathcal{Y}$. Write $c(f+a)=k_*-\ell_*+Y_*$ with obvious notation. Then for every $i=1,\dots, N$
$$U^i(x^i;0)<U^i\big(x^i;c(f^i+a^i)\big)=U^i\big(x^i;k^i_*-\ell^i_*+Y^i_*\big)=U^i\big(x^i;Y^i_*\big)$$
exploiting in the last equality the fact that $k^i_* -\ell^i_* \in \mathcal{C}_i$, and Lemma \ref{l2}.

\end{proof}

\subsection{The continuous time framework}\label{seccontinuous}
In this section, in addition to the tacitly assumed Standing Assumptions and Assumption  \ref{assNew}, we make the following

\begin{assumption}
\label{asscontime}
Let $\mathcal T=[0,T]$ and for each $i$,
    \begin{enumerate}

        \item  
The process $S^i= (S^i_t)_{t \in [0,T] }$ is locally bounded.

        \item We adopt setup \textbf{A}, that is \begin{enumerate}
         \item $W^i=1$.
            \item $L_i=L^{\infty}(\Omega,\mathcal F^i,P^i)=L^{\infty}(\Omega,\mathcal F^i,\Pp)$ and  $L_i^*=L^{1}(\Omega,\mathcal F^i,P^i)$.
        \end{enumerate}
        \item $\mathbb R_0 ^N \subseteq \mcY \subseteq  L^{\infty}(\Omega,\mathbf F,\mathbf P)$.
    \end{enumerate}
\end{assumption}

\begin{remark}[On Assumption \ref{asscontime} ]
As already mentioned, since each $S^i$ is locally bounded, the constant variables $W^i=1$ are both strongly compatible and suitable, which explain the selection in item 2 (a). Similarly, this choice allows us to adopt as ambient space $L_i=L^{\infty}(\Omega,\mathcal F^i,P^i)$ and dual space $L_i^*=L^{1}(\Omega,\mathcal F^i,P^i))$.    
\end{remark}

In this framework
and as $W=1$, the class of admissible trading strategy for the agent $i$ is
\begin{equation}
\mathcal{H}_{i}^1:=\left\{ H\in L(S^i)\mid \exists \alpha>0\text{\mbox{ such that }}%
\int_{0}^{t}H_{s}dS^i_{s}\geq - \alpha , \, \forall t\in \lbrack 0,T]\right\}.
\label{HWi}
\end{equation}%
This is exactly the class of integrable predictable integrands which produce gains processes uniformly bounded from below, as given in \eqref{HW} for $W^i=1$ (i.e. Setup A). Consequently, we choose

\begin{equation}
\mathcal K_{i}=\left\{ \int_0^T H_t dS^i_t\mid H\in \mathcal{H}_{i}^1\right\}.
\label{k_defi}
\end{equation}%
With the selection \eqref{k_defi} and recalling Proposition \ref{propMs}, we conclude that\\
\begin{center}
    \textbf{Assumption \ref{asscontime} implies Assumption \ref{ass12} implies  Assumption \ref{assOK}.}
\end{center}
Thus under Assumption \ref{asscontime} we may apply the results of Sections \ref{secmainres}  and \ref{secharbenefits}. Define

$$\mathcal C_i:=((\mathcal{K}_{i}-L_{+}^{0}(\Omega,\mathcal{F}^i,\probp^i))\cap L^\infty(\Omega,\mathcal{F}^i,\probp^i) ),$$

$$ M_{e,\mathrm{loc}}(S^i)=\{Q \sim P \mid S^i \text{ is a } (\mathcal{F}^i_{t})_{t\in \mathcal T}\text{ -local }  \text{martingale under Q} \}.$$
The classical notion of No Free Lunch for agent $i$ ($\mathbf{NFL}_i$) is

\begin{equation*}
\label{eq:NFLi}
 \mathbf{NFL}_{i}: \text{  } \overline{\mathcal C_i}^\sigma\cap L_{+}^{0}(\Omega, \mathcal{F}^i , P^i)=\{0\}.
 \end{equation*}
 where $\overline{\mathcal{C}_i}^\sigma$ is the $\sigma(L^{\infty}(\Omega, \mathcal{F}^i, P^i),L^{\infty}(\Omega, \mathcal{F}^i , P^i))$- closure of $\mathcal{C}_i$.

 \medskip
In the setting described above, from Theorem 8.2.2 and page 133 by \cite{DS2006} we know that No Free Lunch with Vanishing Risk for agent $i$, denoted by $\mathbf{NFLVR}_i$ implies that $\mathcal C_i$ is weak$^*$ closed and thus $\mathbf{NFLVR}_i$ and  $\mathbf{NFL}_i$ are equivalent conditions. From Theorem 14.1.1 \cite{DS98} it then holds
\begin{theorem}\label{thNFL}
\begin{equation*}
 \mathbf{NFL}_{i} \;\Longleftrightarrow\;  M_{e,\mathrm{loc}}(S^i) \not = \emptyset .
\end{equation*}
\end{theorem}

In the multi--agent segmented market, agents may combine individual trading payoffs $k^i\in \mathcal{K}_i$
with admissible (for example  zero--sum) exchange $Y=(Y^1,\dots,Y^N)\in \mcY$. A {Collective Free Lunch}, as introduced in \cite{F25}, is a nonnegative bounded random vector $f=(f^1,\dots,f^N)\in L^\infty_+(\Omega,\mathbf F,\mathbf P)$,
with $P^j(f^j>0)>0$ for some $j$, that can be attained only \emph{approximately} by feasible collective
positions in the sense of weak$^\ast$ convergence. Recalling the formulation \eqref{NCAKK} of No Collective Arbitrage, namely $\mathcal C^{\mathcal Y} \cap \mathcal L_+ = \{0\}$, for $\mathcal L={\sf{X}}_{i=1}^N L_i$, and considering that in the present setting  $L_i=L^\infty(\Omega,\mathcal  F^i, P^i)$ one formally define

\begin{definition}
 Let $\mathbf P:=(P^1,\dots,P^N)$ and    
\[
\mathcal K^\mcY:={\sf{X}}_{i=1}^N (\mathcal K_i-L^0_+(\Omega,\mathcal F^i,P^i))+\mcY,\qquad \mathcal C^\mcY:=K^\mcY\cap L^\infty(\Omega,\mathbf{F},\mathbf P),
\]
and denote by $\overline{\mathcal C^\mcY}^{\sigma}$ the $\sigma(L^\infty(\Omega,\mathbf F,\mathbf P),L^1(\Omega,\mathbf F,\mathbf P))$--closure of $\mathcal C^\mcY$.
Then \textbf{No Collective Free Lunch} condition ($\mathbf{NCFL(\mathcal{Y})}$)  holds if
\[
\overline{\mathcal C^\mcY}^{\sigma}\cap L^\infty_+(\Omega,\mathbf F,\mathbf P)=\{0\}.
\]
\end{definition}

Intuitively, $\mathbf{NCFL(\mathcal{Y})}$ rules out "sequences" (nets) of bounded collective positions whose
terminal payoffs converge weak$^\ast$ to a sure nonnegative gain for all agents, with strict gain for at
least one agent.

The set of equivalent collective local martingale measures have been already introduced in \eqref{MartingaleLocalMeasures}. Now we recall the following version of the first Collective FTAP, \cite{F25} Theorem 3.7. 
\begin{theorem}
\footnote{ The $\mathbf{NCFL}(\mathcal{Y})$ condition used in \cite{F25} is formulated under a single probability measure, namely $P^1=\dots=P^N$. However, it can be easily shown that allowing for several (possibly different) probability measures, provided that they are all equivalent to one another, yields an equivalent notion of $\mathbf{NCFL}(\mathcal{Y})$. \cite{F25} Theorem 3.7 is more general than the version here stated.} \label{thNFL} Under the Assumption \ref{asscontime}, 
\begin{equation*}
\mathbf{NCFL(\mathcal{Y})} \iff \mathcal M_{e,\mathrm{loc}}(\mcY)  \neq \emptyset,
\end{equation*}
\end{theorem}

 \begin{corollary}[of Theorem \ref{corben}]
 \label{corbenbis}
Under Assumption \ref{asscontime}, the following are equivalent
 \begin{enumerate}
     \item There exists a strictly beneficial $\widehat{Y} \in \mcY$.
     \item There exists a beneficial $\widehat{Y} \in \mcY$.
     \item $\mathbf Q_X \notin \mathcal M_{\mathrm{loc}}(\mcY)$, where $\mathbf Q_X:=(Q_{X^1}^1,\dots,Q_{X^N}^N)$ is the minimax vector and 
     \begin{align*}
\mathcal M_{\mathrm{loc}}(\mcY) &:=\bigg\{ \mathbf Q=(Q^1,\dots,Q^N)  \mid Q^i \in  {M}_{\mathrm{loc}}(S^i) \, \, \text{ and  }   \, \sum_{i=1}^N E_{Q^i}[Y^i] \leq 0 \text {   }\forall Y \in \mcY \bigg\}.
\end{align*}
\end{enumerate}
 \end{corollary}

\begin{proof}
    We already remarked that under Assumption \ref{asscontime} also Assumption \ref{assOK} holds true and, by Proposition \ref{propMs}, $Q^i_{X^i} \in {M}_{\sigma}(S^i)={M}_{\mathrm{loc}}(S^i) $ satisfies $\rn{Q^i_{X^i}}{P} \in L_i^*=L^{1}(\Omega,\mathcal F^i,P^i)$. The thesis now follows from Theorem \ref{corben} noticing that $\mathbf Q_X \notin \mathcal M_{\mathrm{loc}}(\mcY)$ if and only if $\sum_{i=1}^N E_{Q_{X^i}^i}[Y^i]> 0 \text {  for some  } Y \in \mcY.$ 
\end{proof}

Since $\mcY \subseteq L^{\infty}(\Omega,\mathbf F,\mathbf P)$, it follows that  $ \mathcal C^{\mcY}={\sf{X}}_{i=1}^N \mathcal C^i + \mcY$.
As a consequence, the condition $\mathbf{NCFL}(\mcY)$ implies $\mathbf{NFL}_i$ for each agent $i$.
However, it may still occur that a Collective Free Lunch $\mathbf{CFL}(\mcY)$ exists even when each individual market is free of Free Lunch. By Theorem~\ref{thNFL}, this is equivalent to the condition
$\mathcal{M}_{e,\mathrm{loc}}(S^i) \neq \emptyset \quad \text{for all } i.$
This latter condition—implied by the assumption in the next proposition—is therefore the natural one to impose in this setting.

The following proposition shows that the existence of a Collective Free Lunch entails the existence of beneficial exchanges. This result is an immediate consequence of Theorem~\ref{corben}, and will also be established by a direct argument in Section~\ref{secinfoproof}.

\begin{proposition}\label{CFLben}
 Suppose Assumption \ref{asscontime} holds and that, for each $i=1,\dots, N$, $M_{e,\mathrm{loc}} (S^i)\cap \mathcal P_{\Phi^i}(P^i) \neq \emptyset$.  Then the existence of a  Collective Free lunch  $\mathbf{CFL}(\mathcal{Y})$ implies the existence of strict beneficial exchanges.
\end{proposition}
\begin{proof}
The assumption $M_{e,\mathrm{loc}} (S^i)\cap \mathcal P_{\Phi^i}(P^i) \neq \emptyset$ implies, by Proposition \ref{propMs} that the minimizers $Q^i_{X^i} \in {M}_{e,\mathrm{loc}}(S^i) $
   for each $i$. If $\sum_{i=1}^N E_{Q^i_{X^i}}[Y^i] \leq 0 \text {   }\forall Y \in  \mcY $ then $(Q^1_{X^1}, \dots,Q^N_{X^N}) \in \mathcal M_{e,\mathrm{loc}}(\mcY) $, defined in \eqref{MartingaleLocalMeasures}, but this is impossible as $\mathcal M_{e,\mathrm{loc}}(\mcY) 
 =\emptyset $, due to to the existence of a $\mathbf{CFL}(\mathcal{Y})$ and Theorem \ref{thNFL}. Thus $\sum_{i=1}^N E_{Q^i_{X^i}}[Y^i] > 0 \text {  for some  } Y \in \mcY $. Theorem \ref{corben} then imples the thesis.
\end{proof}

\subsubsection{An informative proof of existence of beneficial exchanges if a $\mathbf{CFL}(\mathcal{Y})$ exists}
\label{secinfoproof}
Proposition \ref{CFLben} informs us that the existence of a $\mathbf{CFL}(\mathcal{Y})$ yields existence of beneficial exchanges. The following is a more constructive and informative proof of Proposition \ref{CFLben}. Indeed, it shows how beneficial exchanges can be constructed from a net $w^*$-converging to a $\mathbf{CFL}(\mathcal{Y})$, via  suitable deterministic rebalancing (exchanges in $\R^N_0$).

\begin{lemma}\label{LemmaincreaseBis}
Suppose that for each $i=1,\dots, N$, $M_{e,\mathrm{loc}} (S^i)\cap \mathcal P_{\Phi^i}(P^i) \neq \emptyset$. Under the Assumption \ref{asscontime}, for any  $A\in \mathcal F^i$ with $\Pp (A)>0$, there exists a $C>0$ such that for every $0<c<C$
$$U^i(X^i;c 1_A)>U^i(X^i;0),\quad\text{ for all } i=1,\dots , N.$$
\end{lemma}
\begin{proof}
    The proof is postponed to Appendix \ref{appmix}.
\end{proof}

\begin{proof}[Proof of Proposition \ref{CFLben}]
  We suppose that Assumption \ref{asscontime} holds and that there exists a Collective Free Lunch, namely:
 $$\exists f\in \overline{\mathcal C^\mathcal{Y}}^\sigma\text{ s.t. } f^i\geq 0\,\, \forall i\text{ and }P^j(f^j>0)>0\text{ for some }j$$
 where $\overline{\mathcal C^{\mathcal{Y}}}^\sigma$ stands for the $\sigma(L^\infty(\Omega,\mathbf{F}, \mathbf {P}),L^1(\Omega,\mathbf{F}, \mathbf {P}))$-closure of $\mathcal{C}^{\mcY}$.
 We show that there exists
$Y_* \in \mathcal Y $ such that
 $U^i(X^i;{Y_*^i}) > U^i(X^i;0) \, \text{ for all } i=1,\dots,N.$
Observe that since $\mathcal{C}^\mathcal{Y}$ is a convex cone its weak$^*$ closure coincides with its closure in the Mackey topology: $\overline{\mathcal C^\mathcal{Y}}^\sigma=\overline{\mathcal C^\mathcal{Y}}^\tau$ for $\tau=\tau(L^\infty(\Omega,\mathbf{F}, \mathbf {P}),L^1(\Omega,\mathbf{F}, \mathbf {P}))$ said Mackey topology.
Take $A_n:=\{f^j\geq \frac1n\}$. For some $N$ big enough we have $P^j(A_N)>0$. For the given set $A_N$ pick $C$ from Lemma \ref{LemmaincreaseBis}. 
The properties of $f$ ensure that
$$U^i(X^i;cf^i)\geq U^i(X^i;0)\,\forall i\text{ and } U^j(X^j;cf^j)\geq U^j\bigg(X^j;\frac{c}{N}1_{A_N}\bigg)>U^j(X^j;0)$$
using Lemma \ref{LemmaincreaseBis} for the last inequality.
Under Assumption \ref{asscontime} we have (from Proposition \ref{propU*}) finiteness and norm continuity of $U^i(X^i;\cdot)$ on $L^\infty(\Omega, \mathcal{F}^i,P^i)$. In particular, we can pick $a\in \R^N_0$ with $a^j<0$ and $a^i>0$
 for every $i\neq j$ in such a way that 
$$U^i\big(X^i;c(f^i+a^i)\big)> U^i(X^i;0)\,\forall i.$$
Note that $c(f+a)\in \overline{\mathcal{C}^\mathcal{Y}}^\sigma=\overline{\mathcal{C}^\mathcal{Y}}^\tau$ since $a\in \R^N_0\subseteq \mathcal{Y}$ and the closure of a convex cone is a convex cone.
Let then $k_\alpha+Y_\alpha-\ell_\alpha\in \mathcal{C}^\mathcal{Y}$ define a net converging in $\tau$ to $c(f+a)$.
We observe that for every $\alpha$ and $i=1,\dots,N$
$$ U^i(X^i; k^i_\alpha+Y^i_\alpha-\ell^i_\alpha)=U^i(X^i; Y^i_\alpha).$$ Indeed,  $Y_\alpha\in \mcY \subseteq L^\infty(\Omega,\mathbf{F}, \mathbf {P})$ and hence if $k_\alpha+Y_\alpha-\ell_\alpha\in\mathcal{C}^\mathcal{Y}$ then $k^i_\alpha-\ell^i_\alpha\in\mathcal{C}_i$. Now we see 
\begin{align*}
 \lim_\alpha U^i(X^i;Y^i_\alpha)&=\lim_\alpha U^i(X^i;k^i_\alpha+Y^i_\alpha-\ell^i_\alpha) =U^i(X^i;c(f^i+a^i))>U^i(X^i;0).
\end{align*}
where the second equality follows from the Mackey continuity of indirect utility (Proposition \ref{mackeycont}). 
Then, there exists $\beta$ such that $\alpha \succeq \beta$ implies 
$U^i\big(X^i;Y^i_\alpha\big) >U^i(X^i;0)$ for $i=1,\dots,N$. In particular, $Y_*=Y_\beta$ behaves as desired.  
\end{proof}

\begin{table}[ht]
\centering
\caption{Summary of core implications regarding Collective Arbitrage (CA) and beneficial claims.}
\label{tab:implications}
\renewcommand{\arraystretch}{1.5}
\begin{tabular}{|l|p{0.62\textwidth}|}
\hline
\textbf{Market Condition} & \textbf{Implication} \\ \hline
Existence of a CA & Existence of a beneficial $Y \in \mathcal{Y}$ \\ \hline
Existence of a CFL & Existence of a beneficial $Y \in \mathcal{Y}$  \\ \hline
No Collective Arbitrage (NCA) & A beneficial $Y \in \mathcal{Y}$ may  or may not exist \\ \hline
No Collective Free Lunch (NCFL) & A beneficial $Y \in \mathcal{Y}$ may  or may not exist \\ \hline
$\mathbf Q_X \notin \mathcal{M}_{e,\sigma}(\mcY)$  and NCFL & Beneficial $Y \in \mathcal{Y}$ exist \\ \hline
Completeness of single agents' markets & If NCA holds, then no beneficial $Y \in \mathcal{Y}$ exists \\ \hline
\end{tabular}
\end{table}

\section{Example: discrete time}
\label{secexamplesdiscr}

We provide an explicit example where: dim$(\mcY) =3$, \textbf{NCA}$(\mcY)$ holds true and the market is $\mcY$-collectively incomplete ($ \abs{\Me} =\infty) $, there exists a global arbitrage ($M_e(S^1) \cap M_e(S^2) = \emptyset$),  dim$(\widehat{\mcY})=2$ and $|\mathcal{M}_e(\widehat{\mcY})|=\infty$. 
We consider a two periods market model, $\mathcal T=\{0,1,2\}$,  with two agents and two stocks and where each agent $i$, $i=1,2$, may invest only in the stock $X^i$ and in the riskless asset $X^0_t=1$ for all $t\in \mathcal T$. 
The evolution of the price processes is described in Figure \ref{figtree2}. We take $|\Omega|=8$, a common filtration: $\mcF_0=\{\emptyset, \Omega\},\mcF_2=\mathcal{P}(\Omega) $ and $$ \mcF_1=\sigma(A_1=\{\omega_1,\omega_2\},A_2=\{\omega_3,\omega_4\},A_3=\{\omega_5,\omega_6\}, A_4=\{\omega_7,\omega_8\}
).$$
We also select 
$$\mathcal{Y} = \left\{ Y \in L^0(\Omega, \mathbf{F}_1, P) \mid \sum_{i=1}^N Y^i = 0  \right\}.$$
With this choice, the set of collective equivalent martingale measures becomes
 \begin{align*}
\mathcal M_{e,\mathrm{loc}}(\mcY)=\mathcal M_{e}(\mcY) &:=\bigg\{ \mathbf Q=(Q^1,Q^2)  \mid Q^i \in  {M}_{e}(S^i) \, \text{ and }   \,  E_{Q^1}[Y^1]+E_{Q^2}[Y^2] \leq 0 \text {   }\forall Y \in \mcY \bigg\}\\
&=\bigg\{ \mathbf Q=(Q^1,Q^2)  \mid Q^i \in  {M}_{e}(S^i) \, \text{ s.t. } Q^1=Q^2 \text{ on } \mathcal F_1 \bigg\}.
\end{align*}
In the following, we
express measures $Q$ on $\mcF_1$ as $4$-tuples, in the form $(Q(A_1),Q(A_2),Q(A_3),Q(A_4))$.

In Figure \ref{figtree2} the parameters $q,q',p,p'$ parametrizing the sets of equivalent martingale measures for the two stocks satisfy:
\begin{equation}
\label{eqconstraint}
    0<q<\frac12,\,0<q'<\frac12-q\quad\text{and}\quad 0<p<\frac12,\,p<p'<\frac45-\frac35p
\end{equation}


    Denoting by $M^i_e|_{\mcF_1}$ the collection of restrictions to $\mcF_1$ of elements in $M_e(S^i)={M}_{e, \mathrm{loc}}(S^i)$, one readily verifies  that
\begin{equation}\label{mesonf1}
    \begin{split}
     M^1_e|_{\mcF_1}&=\bigg\{\bigg(\frac12, q,q', \frac12-(q+q')\bigg),0<q<\frac12,\,0<q'<\frac12-q \bigg\},\\
    M^2_e|_{\mcF_2}&=\bigg\{\bigg(p', p,\frac14(p'-p), 1-\Big(\frac54 p'+\frac34 p\Big)\bigg), 0<p<\frac12,\,p<p'<\frac45-\frac35p\bigg\},   
    \end{split}
\end{equation}
and that the set of restrictions to $\mcF_1$ of elements in $\Me $  is
\begin{equation}\label{MeMe}
\Me |_{\mcF_1}=\bigg\{\bigg(\Big(\frac12, q,\frac18-\frac14 q, \frac38-\frac34q\Big),\Big(\frac12, q,\frac18-\frac14 q, \frac38-\frac34q\Big)\bigg), 0<q<\frac12\bigg\}.
\end{equation}

\begin{figure}
\begin{center}
\tikzstyle{level 1}=[level distance=2cm, sibling distance=2.5cm,->]
\tikzstyle{level 2}=[level distance=1.5cm, sibling distance=1cm,->]

\tikzstyle{bag} = [text width=1.5em, text centered]
\tikzstyle{end} = []

\begin{tikzpicture}[grow=right, sloped]
\node[bag](c1){$8$}
    child {
        node[bag]{$4$}        
            child {
                node[end, label=right:
                    {$2$}](y18) {}
                edge from parent
                node[above] {}
                node[below]  {$\blue{1/2}$}
            }
            child {
                node[end, label=right:
                    {$6$}](y17) {}
                edge from parent
                node[above] {$\blue{1/2}$}
                node[below]  {}
            }
            edge from parent 
            node[above] {}
            node[below]  {$\blue{(1/2)-(q+q')}$}
    }
    child {
        node[bag]{$4$}        
            child {
                node[end, label=right:
                    {$3$}](y16) {}
                edge from parent
                node[above] {}
                node[below]  {$\blue{1/2}$}
            }
            child {
                node[end, label=right:
                    {$5$}](y15) {}
                edge from parent
                node[above] {$\blue{1/2}$}
                node[below]  {}
            }
            edge from parent 
            node[above] {$\blue{q'}$}
            node[below]  {}
    }
    child {
        node[bag]{$4$}        
            child {
                node[end, label=right:
                    {$2$}](y14) {}
                edge from parent
                node[above] {}
                node[below]  {$\blue{1/2}$}
            }
            child {
                node[end, label=right:
                    {$6$}](y13) {}
                edge from parent
                node[above] {$\blue{1/2}$}
                node[below]  {}
            }
            edge from parent 
            node[above] {$\blue{q}$}
            node[below]  {}
    }
    child {
        node[bag] {$12$}        
        child {
                node[end, label=right:
                    {$8$}] (y12){}
                edge from parent
                node[above] {}
                node[below]  {$\blue{3/4}$}
            }
            child {
                node[end, label=right:
                    {$24$}](y11) {}
                edge from parent
                node[above] {$\blue{1/4}$}
                node[below]  {}
            }
        edge from parent         
            node[above] {$\blue{1/2}$}
            node[below]  {}
    };

\node[bag](y21) at ([xshift=1.1cm]y11) {$\red{\omega_1}$};
\node[bag](y22) at ([xshift=1.1cm]y12) {$\red{\omega_2}$};
\node[bag](y23) at ([xshift=1.1cm]y13) {$\red{\omega_3}$};
\node[bag](y24) at ([xshift=1.1cm]y14) {$\red{\omega_4}$};
\node[bag](y25) at ([xshift=1.1cm]y15) {$\red{\omega_5}$};
\node[bag](y26) at ([xshift=1.1cm]y16) {$\red{\omega_6}$};
\node[bag](y27) at ([xshift=1.1cm]y17) {$\red{\omega_7}$};
\node[bag](y28) at ([xshift=1.1cm]y18) {$\red{\omega_8}$};
\node[bag](c2) at ([xshift=6cm]c1){$20$}
    child {
        node[bag]{$20$}        
            child {
                node[end, label=right:
                    {$16$}](z18) {}
                edge from parent
                node[above] {}
                node[below]  {$\blue{1/2}$}
            }
            child {
                node[end, label=right:
                    {$24$}](z17) {}
                edge from parent
                node[above] {$\blue{1/2}$}
                node[below]  {}
            }
            edge from parent 
            node[above] {}
            node[below]  {$\blue{1-((5/4)p'+(3/4)p)}$}
    }
    child {
        node[bag]{$4$}        
            child {
                node[end, label=right:
                    {$2$}](z16) {}
                edge from parent
                node[above] {}
                node[below]  {$\blue{1/2}$}
            }
            child {
                node[end, label=right:
                    {$6$}](z15) {}
                edge from parent
                node[above] {$\blue{1/2}$}
                node[below]  {}
            }
            edge from parent 
            node[above] {$\blue{(1/4)(p'-p)}$}
            node[below]  {}
    }
    child {
        node[bag]{$16$}        
            child {
                node[end, label=right:
                    {$12$}](z14) {}
                edge from parent
                node[above] {}
                node[below]  {$\blue{1/2}$}
            }
            child {
                node[end, label=right:
                    {$20$}](z13) {}
                edge from parent
                node[above] {$\blue{1/2}$}
                node[below]  {}
            }
            edge from parent 
            node[above] {$\blue{p}$}
            node[below]  {}
    }
    child {
        node[bag] {$24$}        
        child {
                node[end, label=right:
                    {$48$}] (z12){}
                edge from parent
                node[above] {}
                node[below]  {$\blue{1/4}$}
            }
            child {
                node[end, label=right:
                    {$16$}](z11) {}
                edge from parent
                node[above] {$\blue{3/4}$}
                node[below]  {}
            }
        edge from parent         
            node[above] {$\blue{p'}$}
            node[below]  {}
    };

\node[bag](z21) at ([xshift=1.1cm]z11) {$\red{\omega_1}$};
\node[bag](z22) at ([xshift=1.1cm]z12) {$\red{\omega_2}$};
\node[bag](z23) at ([xshift=1.1cm]z13) {$\red{\omega_3}$};
\node[bag](z24) at ([xshift=1.1cm]z14) {$\red{\omega_4}$};
\node[bag](z25) at ([xshift=1.1cm]z15) {$\red{\omega_5}$};
\node[bag](z26) at ([xshift=1.1cm]z16) {$\red{\omega_6}$};
\node[bag](z27) at ([xshift=1.1cm]z17) {$\red{\omega_7}$};
\node[bag](z28) at ([xshift=1.1cm]z18) {$\red{\omega_8}$};

\end{tikzpicture}
\end{center}
\caption{Tree for the stocks $(X^1,X^2)$ at times $t=0,1,2$.}
\label{figtree2}
\end{figure}
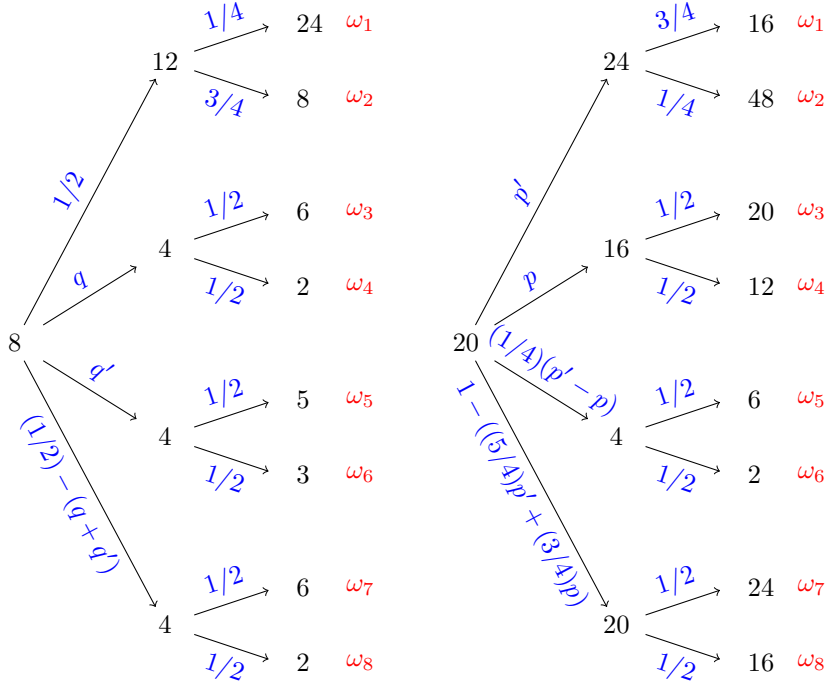

For  each agent \(i=1,2\), we consider  the quadratic utility functions
\[
u^i(x)= x^2 \mathbf 1_{\{x\le 0\}}
\]
and deterministic initial endowments $x^i<0$. The convex conjugate of \(u_i\) is $\Phi_i(y)=
\sup_{x\in\mathbb R}\{u_i(x)-xy\}
=\dfrac{y^2}{4\gamma_i}$, for $ y \geq 0$,
and the dual formulation of the optima wealth becomes
\[
U^i(x^i;0)=\min_{\lambda>0}\min_{Q\in M_e(S^i)}
\left\{
\lambda x^i+\frac{\lambda^2}{4\gamma_i}
E_{P_i}\!\left[\left(\frac{dQ}{dP^i}\right)^2\right]
\right\}=\min_{Q\in M_e(S^i)} -\frac{\gamma_i (x^i)^2}{E_{P^i}\!\left[\left(\frac{dQ}{dP^i}\right)^2\right]}.
\]
Therefore the optimizer $Q^i_{x^i} $ of the dual problem is the well known minimal variance martingale measure. 
To simplyfy, we assume that the subjective probabilities are uniform:
\[
P^i(\{\omega_k\})=\frac18,\qquad k=1,\dots,8.
\]
Direct computation of the minimax measures is possible using Lagrange multipliers. The unique minimax measures $Q^i_{x^i}$ are given by
\begin{align*}
    Q^1_{x^1}&=
\left(
\frac18,\frac38,\frac1{12},\frac1{12},\frac1{12},\frac1{12},\frac1{12},\frac1{12}
\right),\\
Q^2_{x^2}
&=
\left(
\frac{33}{125},\frac{11}{125},
\frac{17}{125},\frac{17}{125},
\frac{1}{100},\frac{1}{100},
\frac{89}{500},\frac{89}{500}
\right)
\end{align*}
so that \[
Q^1_{x^1}\big|_{\mathcal F_1}
=
\left(
\frac{1}{2},\frac{1}{6},\frac{1}{6},\frac{1}{6}
\right)\in M^1_e|_{\mcF_1}, \quad 
Q^2_{x^2}\big|_{\mathcal F_1}
=
\left(
\frac{44}{125},\frac{34}{125},\frac{1}{50},\frac{89}{250}
\right)\in M^2_e|_{\mcF_1}.
\]
We immediately note that $\mathbf Q_X \notin \mathcal M_{\mathrm{loc}}(\mcY)$, so that Corollary \ref{corbentris} ensures existence of striclty beneficial exchanges.
To exhibit these explicitly, we now follow the proof of Theorem \ref{corben}, implication 3.  $\Rightarrow $ 1. We note that $Y=(1_{A_1},-1_{A_1})\in\mcY$ satisfies
\[
E_{Q^1_{x^1}}[Y]+E_{Q^2_{x^2}}[-Y]
=
\frac12-\frac{44}{125}
=
\frac{37}{250}>0.
\]
We define \(\widehat Y=(\widehat Y^1,\widehat Y^2)\) as
\[
\widehat Y^1
=
Y+\frac12\left(E_{Q^1_{x^1}}[Y^1]+E_{Q^2_{x^2}}[-Y^2]\right)-E_{Q^1_{x^1}}[Y^1],
\]
\[
\widehat Y^2
=
-Y+\frac12\left(E_{Q_{x^1}}[Y^1]+E_{Q_{x^2}}[-Y^2]\right)-E_{Q_{x^2}}[-Y^2].
\]
Then 
\[
E_{Q_{x^1}}[\widehat Y^1]=
\frac{37}{500}>0, \quad 
E_{Q_{x^2}}[\widehat Y^2]
=
\frac{37}{500}>0.
\]
By the same argument in the proof of Theorem \ref{corben}, implication \(3.\Rightarrow 1.\), we have that $\alpha \widehat{Y}\in\mcY$ is strictly beneficial for $\alpha>0$ sufficiently small.
\section{Example: continuous time}
\label{secexamplescont}

In this example we work in the general semimartingale setup depicted in \cite{BF05}. We consider $N$ agents endowed with CARA utilities $u^i(x)=-\frac{1}{\gamma^i}e^{-\gamma^ix}$ and subjective probability $P^i$. Since the exponential utilities satisfies the Reasonable Asymptotic Elasticity, then Assumption 2 in \cite{BF05} is automatically met (see also \cite{BF05}, Section 2.2).  Moreover we assume for any $i=1,\ldots,N$, the existence of suitable and strongly compatible $W^i$,  (as in Definitions \ref{viability} and \ref{compatibility}), such that $U^i_{W^i}(x^i_i;0)<u^i(+\infty)$, where
\begin{equation*}
U^i_{W^i}(x^i;0):=\sup_{k^i\in K_i^{W^i}}E_{P^i}[u^i(x^i+k^i)] \text{ for } K_i^{W^i} \text{ defined in  \eqref{k_defi}.} 
\end{equation*}
Then \cite{BF05}, Theorem 1 guarantees the existence of  $k^i_{\star}\in K_i^{W^i}$ such that $U^i_{W^i}(x^i;0)= E_{P^i}[u^i(x^i+k^i_{\star})]$ for all $ i=1,\dots, N$, with $U^i_{W^i}(x^i)$ which does not depend on the choice of $W^i$,  among the suitable ones. Moreover the optimal profile can be explicitly represented as
\[k^i_{\star}= -x^i-\Phi^{\prime}_i\left(\lambda_{x^i}\rn{Q_{x^i}}{P^i}\right)\]
where $\lambda_{x^i}>0,Q_{x^i}$ are the dual solutions computed in \eqref{dual14} and $\Phi^{\prime}_i(y)=\frac{1}{\gamma^i}\ln{y}$.

\medskip

If we now take into consideration an admissible exchange $Y\in\mcY$ we have

\[U^i_{W^i}(x^i;Y^i)=\sup_{k^i\in K_i^{W^i}}E_{P^i}[u^i(x^i+k^i+Y^i)]\geq E_{P^i}[u^i(x^i+k^i_*+Y^i)].\]

\noindent For any $i=1,\ldots,N$, we can observe 
\[\lambda_{x^i}\cdot \rn{Q_{x^i}}{P^i}=\exp\left({\gamma^i\Phi^{\prime}_i\left(\lambda_{x^i}\rn{Q_{x^i}}{P^i}\right)}\right)\]
and consider the equivalent change of measure 

\[\frac{dQ_{x^i}}{dP^i}=\frac{1}{\lambda_{x^i}}e^{-\gamma^i(x^i+k^i_{\star})}\]
obtaining in this way 
\[U^i_{W^i}(x^i;Y^i)\geq E_{Q_{x^i}}\left[\lambda_{x^i}\cdot u^i(Y^i)\right] \text{ and } U^i_{W^i}(x^i;0)=E_{Q_{x^i}}\left[\lambda_{x^i}\cdot u^i(0)\right].\]
The procedure has reduced the problem to the search of beneficial exchanges without a financial market. Therefore we can invoke Corollary \ref{prop3conditions}, adopting the new measures $Q_{x^i}$ (in the place of $P^i$). Indeed if 
$\sum_{i=1}^N E_{Q_{x^i}}[\tilde{Y}^i]> 0$ for some $\tilde{Y} \in \mcY$ then we can find an other $Y\in\mcY$ such that
\[E_{Q_{x^i}}[u^i(Y^i)]> E_{Q_{x^i}}[u^i(0)]\]
for all $i \in \{1,\dots,N\}$, which implies \(U^i_{W^i}(x^i;Y^i)>U^i_{W^i}(x^i;0)\) for all $i \in \{1,\dots,N\}$.
\\A simple explicit example can be built up considering $N=2$: if there exists $A\in \mathcal{F}$ such that $Q_{x^1}(A)>Q_{x^2}(A)$, then the zero sum exchange $\tilde{Y}^1= 1_A$, $\tilde{Y}^2=- 1_A$ satisfies $E_{Q_{x^i}}[\tilde{Y}^1]+E_{Q_{x^i}}[\tilde{Y}^2]> 0$.
Moreover explicit computations can explain how the discrepancy between $Q_{x^1}(A)$ and $Q_{x^2}(A)$ can practically create a benefit for agents $i,j$. In particular for any  $(a,b)$ chosen in the set
\( \{(\alpha,\beta):0<\alpha<\alpha^\ast,\;L(\alpha)<\beta<U(\alpha)\}
\), 
where
\[
U(\alpha)=-\frac{1}{\gamma_1}\ln\!\big(q_1 e^{-\gamma_1\alpha}+1-q_1\big),\qquad
L(\alpha)=\frac{1}{\gamma_2}\ln\!\big(q_2 e^{\gamma_2\alpha}+1-q_2\big).
\]
and $\alpha^\ast>0$ is the unique root of $U(\alpha)=L(\alpha)$, we indeed have 

\[U^i_{W^i}(x^i;a 1_A-b)> U^i_{W^i}(x^i;0) \text{ and } U^i_{W^j}(x^j;-a 1_A+b)> U^j_{W^j}(x^j;0)
\]

as a consequence of the following Lemma.
\begin{lemma}
Let  and assume that for two agents (w.l.o.g. for $i=1,2$)   $q_1:=Q_{x^1}(A)>Q_{x^2}=:q_2(A)$, with $A\in\mathcal{F}$. If $Y=(1_A,-1_A)\in\mathcal{Y}$, then the set of solutions  $\alpha,\beta\in (0,\infty)$ such that 
\[E_{Q_{x^1}}[\exp{(-\gamma^1(\alpha 1_A-\beta))}]< 1 \text{ and } E_{Q_{x^2}}[\exp{(-\gamma^2(-\alpha 1_A+\beta))}]< 1\]
is
\( \{(\alpha,\beta):0<\alpha<\alpha^\ast,\;L(\alpha)<\beta<U(\alpha)\},
\)
where
\[
U(\alpha)=-\frac{1}{\gamma_1}\ln\!\big(q_1 e^{-\gamma_1\alpha}+1-q_1\big),\qquad
L(\alpha)=\frac{1}{\gamma_2}\ln\!\big(q_2 e^{\gamma_2\alpha}+1-q_2\big).
\]
and $\alpha^\ast>0$ is the unique solution of $U(\alpha)=L(\alpha)$.

\end{lemma}

\begin{proof}

The two strict inequalities are equivalent to
\[
e^{\gamma_1\beta}\big(q_1 e^{-\gamma_1\alpha}+(1-q_1)\big)<1
\quad\Longleftrightarrow\quad
\beta<U(\alpha):=-\frac{1}{\gamma_1}\ln\!\big(q_1 e^{-\gamma_1\alpha}+1-q_1\big),
\]
and
\[
e^{-\gamma_2\beta}\big(q_2 e^{\gamma_2\alpha}+(1-q_2)\big)<1
\quad\Longleftrightarrow\quad
\beta>L(\alpha):=\frac{1}{\gamma_2}\ln\!\big(q_2 e^{\gamma_2\alpha}+1-q_2\big).
\]
Thus the feasible pairs must satisfy
\[
L(\alpha)<\beta< U(\alpha).
\]
For \(\alpha>0\) one has
\[
U'(\alpha)=\frac{q_1 e^{-\gamma_1\alpha}}{q_1 e^{-\gamma_1\alpha}+1-q_1}>0,
\qquad
L'(\alpha)=\frac{q_2 e^{\gamma_2\alpha}}{q_2 e^{\gamma_2\alpha}+1-q_2}>0,
\]
so both \(U\) and \(L\) are strictly increasing on \((0,\infty)\).
Their endpoint behaviour is
\[
U(0^+)=0,\qquad U(+\infty)=-\frac{1}{\gamma_1}\ln(1-q_1)\in(0,\infty),
\]
\[
L(0^+)=0,\qquad L(+\infty)=+\infty.
\]
Define
\[
h(\alpha):=U(\alpha)-L(\alpha).
\]
We study $h'(\alpha)$ for $\alpha>0$. A simple computation yields
\[
h'(\alpha)=U'(\alpha)-L'(\alpha)
=
\frac{
q_1(1-q_2)e^{-\gamma_1\alpha}
-
q_2(1-q_1)e^{\gamma_2\alpha}
}{
\big(q_1 e^{-\gamma_1\alpha}+1-q_1\big)
\big(q_2 e^{\gamma_2\alpha}+1-q_2\big)
}
\]
so that for $\alpha>0$ we have $h'(\alpha)>0$ if and only if $\alpha<
\frac{1}{\gamma_1+\gamma_2}
\ln\left(\frac{q_1(1-q_2)}{q_2(1-q_1)}\right)$ (note that since $1-q_1<1-q_2$, the latter quantity is strictly positive). We also have \(h(0)=0\), \(h'(0)=q_1-q_2>0\) and $h(+\infty)=-\infty$. 

Therefore there exists a unique
\(\alpha^\ast>0\) solving
\[
U(\alpha^\ast)=L(\alpha^\ast),
\]
equivalently
\[
-\frac{1}{\gamma_1}\ln\!\big(q_1 e^{-\gamma_1\alpha^\ast}+1-q_1\big)
=\frac{1}{\gamma_2}\ln\!\big(q_2 e^{\gamma_2\alpha^\ast}+1-q_2\big).
\]
For every \(\alpha\in(0,\alpha^\ast)\) we have \(U(\alpha)>L(\alpha)>0\), hence the set of solutions is
\[
\big\{(\alpha,\beta):\ 0<\alpha<\alpha^\ast,\; L(\alpha)<\beta< U(\alpha)\big\}.
\]

\end{proof}

\appendix
\section{Appendix: miscellaneous results}
\label{appmix}
The utility function $u$ in the following lemma is not assumed to be strictly concave on $(a,\infty)$. For this reason, standard results on convex conjugates from the literature do not apply directly. For completeness, we therefore provide an elementary proof.\\

\begin{lemma}\label{lemmauphi} 
Let $u:\R\rightarrow \R$ be concave, nondecreasing, differentiable at every $x\in\R$ and satisfying the Inada conditions. 
Suppose also that $u$ is strictly concave on \((-\infty,a)\), where $a:=\inf \{x \in \R : u(x)=u(+\infty) \} \leq +\infty$.
For $y\in\mathbb{R}$ let
$$\Phi(y):=\sup_{x\in\mathbb{R}}\big(u(x)-xy\big)=-u^*(y),\quad \quad \mathrm{dom}(\Phi):=\{y \in \R \mid \Phi(y)<+\infty \}.
$$
Then $(0,+\infty) \subseteq \mathrm{dom}(\Phi) \subseteq [0,+\infty)$, $\Phi$ is bounded from below  and it is strictly convex on $(0,+\infty) $.  
\end{lemma}

\ifLONG
\begin{proof}
Note that $u$ satisfies the Inada limits
\[
\lim_{x\to-\infty}u'(x)=+\infty, \qquad \lim_{x\to+\infty}u'(x)=0.
\]

Since \(u\) is concave and differentiable, \(u'\) is nonincreasing on \(\mathbb R\).
On \((-\infty,a)\), strict concavity implies that \(u'\) is strictly decreasing.
Moreover, since \(u\) is differentiable, \(u'\) has the Darboux property.
A strictly monotone function with the intermediate value property is continuous;
hence \(u'\) is continuous on \((-\infty,a)\).

Fix \(y>0\) and consider \(g_y(x):=u(x)-xy\).
Then \(g_y'(x)=u'(x)-y\).
By the Inada limits,
\[
\lim_{x\to-\infty}g_y'(x)=+\infty, \qquad
\lim_{x\to+\infty}g_y'(x)=-y<0.
\]
Therefore \(g_y'\) vanishes at least once.
If \(x\ge a\), then \(u'(x)=0\) since \(u\) is constant on \([a,+\infty)\); hence
\(g_y'(x)<0\) for all \(x\ge a\).
Thus any critical point of \(g_y\) must lie in \((-\infty,a)\).
On this interval \(u'\) is strictly decreasing, so the equation \(u'(x)=y\) has a unique solution \(x^*(y)\in(-\infty,a)\).
Consequently the supremum in the definition of \(\Phi(y)\) is attained uniquely at \(x^*(y)\).

Let \(y_n\to y>0\) and set \(x_n:=x^*(y_n)\).
Choose \(\delta>0\) such that \([y-\delta,y+\delta]\subset(0,\infty)\).
By the Inada limits and continuity of \(u'\) on \((-\infty,a)\), the set
\(\{x\in\mathbb R:\ u'(x)\in[y-\delta,y+\delta]\}\) is compact.
For all sufficiently large \(n\), \(y_n\in[y-\delta,y+\delta]\), hence \(x_n\) belongs to this compact set.
Let \(x_{n_k}\to\bar x\) be a convergent subsequence.
Continuity of \(u'\) yields
\[
u'(\bar x)=\lim_{k\to\infty}u'(x_{n_k})=\lim_{k\to\infty}y_{n_k}=y.
\]
By uniqueness of the solution to \(u'(x)=y\) in \((-\infty,a)\), it follows that \(\bar x=x^*(y)\).
Hence \(x_n\to x^*(y)\), and the map \(y\mapsto x^*(y)\) is continuous on \((0,\infty)\).

For \(h\) small, using the maximizing property of \(x^*(y)\) and \(x^*(y+h)\),
\[
\Phi(y+h)-\Phi(y)\ge -h\,x^*(y), \qquad
\Phi(y+h)-\Phi(y)\le -h\,x^*(y+h).
\]
Dividing by \(h\) and letting \(h\to0\), the continuity of \(x^*(\cdot)\) implies
\[
\Phi'(y)=-x^*(y).
\]
Thus \(\Phi\) is differentiable on \((0,\infty)\).
Fix \(0<y_1<y_2\). Then \(u'(x^*(y_i))=y_i\) for \(i=1,2\). Since \(u'\) is strictly decreasing on \((-\infty,a)\), from
\(y_1<y_2\) we infer \(x^*(y_1)>x^*(y_2)\). Hence
\[
\Phi'(y_1)=-x^*(y_1)<-x^*(y_2)=\Phi'(y_2).
\]
Thus \(\Phi'\) is strictly increasing on \((0,\infty)\). A convex function whose derivative is strictly increasing is strictly convex; therefore \(\Phi\) is strictly convex on \((0,\infty)\).
\end{proof}
\fi

\begin{lemma}
\label{kotheprop}
    Under Assumption \ref{assOK} we have for every $i=1,\dots,N$:
 \begin{enumerate}
     \item $L_i, L_i^*\subseteq L^1(\Omega,\mathcal{F}^i,\probp^i)$ and both are decomposable (under $P^i$) in the sense of \cite{Rockafellar} page 532;
     \item With an obvious abuse of notation,  $L_i^*$ coincides with the set of $\sigma-$order continuous linear functionals\footnote{i.e. linear functionals $\phi:L_i\rightarrow \R$ such that $\phi(X_n)\rightarrow 0$ for every sequence $(X_n)_n$ in $L_i$ such that $X_n\rightarrow 0$ a.s. and $\sup_n\abs{X_n}\in L_i$} on $L_i$. To be more precise, a linear functional $\phi$ on $L_i$
 is $\sigma$-order continuous if and only if there exists a (unique) $Z\in L_i^*$ s.t. $\phi(X)=\Epo{XZ}$ for every $X\in L_i$, and $\phi$ is positive if and only if $Z_i\geq 0$ a.s.;
 \item A linear functional $\phi$ on $L_i$ is $\sigma$-order continuous if and only if it satisfies $\phi(X_n)\rightarrow 0$ for every sequence $(X_n)_n$ in $L_i$ such that $X_n\geq 0$ for every $n$, $X_n\rightarrow 0$ a.s. and $\sup_n\abs{X_n}\in L_i$;
 \item  $L_i^*\subseteq L'_i$, the latter being the topological dual of $L_i$.
 \end{enumerate}   
\end{lemma}
\begin{proof}
We start with Item 1. We note that since $L^\infty(\Omega,\mathcal{F}^i,P^i)=L^{\infty}(\Omega,\mathcal{F}^i,\red{\Pp})\subseteq L_i$, we have $L_i^*\subseteq L^1(\Omega,\mathcal{F}^i,P^i)\cap L^1(\Omega,\mathcal{F}^i,\red{\Pp})$. As already noted after Assumption \ref{assOK} was stated, 
$L_i\subseteq L^1(\Omega,\mathcal{F}^i,\probp^i)$.
To verify decomposability we first verify that $L^\infty(\Omega,\mathcal{F}^i,\probp^i)\subseteq L_i\cap L_i^*$. $L^\infty(\Omega,\mathcal{F}^i,\probp^i)\subseteq L_i$ is assumed already, while $L_i\subseteq L^1(\Omega,\mathcal{F}^i,\probp^i)$ implies $L^\infty(\Omega,\mathcal{F}^i,\probp^i)\subseteq L_i^*$ by definition of the latter. We conclude by observing that whenever $E\in \mathcal{F}^i $ is given, clearly $\abs{1_E X}\leq \abs{X}$ for every $X\in L_i$, so that by Assumpton \ref{assOK} we have $1_EX\in L_i$. Furthermore for every $Z\in L^*_i$ we have $\Epo{\abs{X (1_E) Z}}\leq \Epo{\abs{XZ}}<+\infty$ so that $1_E Z\in L^*_i$. These observations show decomposability.
Regarding Item 3, we note that the definition of $\sigma$-order continuity in Item 2 ensures that $\phi(X_n)\rightarrow 0$ for every dominated (in $L_i$) nonnegative sequence converging a.s. to zero. To prove the implication ($\Leftarrow$), for a general dominated sequence $(X_n)$ converging to zero, the sequences $(X^\pm_n)_n$ of positive and negative parts are both nonnegative, converge to zero and are still dominated in $L_i$. Hence $\phi(X_n^\pm)\rightarrow 0$. One then concludes $\phi(X_n)\rightarrow 0$ by linearity.
We come to Items 2 and 4.
    Since $1_E\in L^\infty(\Omega,\mathcal{F}^i,\probp^i)\subseteq L_i$ for every $E\in \mathcal{F}^i$ we have that $\mathrm{supp}(L_i)=\Omega$ (see \cite{Kanto} Chapter IV Section 3), and this together with the other requirements on $L_i$ imply that $L_i$ is a Banach foundation space, hence a Banach ideal space (BIS) in the notation of  \cite{Kanto} Chapter IV Section 3. \cite{Kanto} Theorems IV.3.1  and IV.3.3 yield the desired results. 
\end{proof}

\begin{proof}[Proof of Lemma \ref{LemmaincreaseBis}]
 The conditions in Assumption \ref{asscontime} ensure that the ones in Assumption \ref{assOK} are all satisfied. From \eqref{dualrho1tris}, we see that  for any $Z \in L_i$ with $U^i(X^i;Z)<u^i(+\infty)$
 \begin{align}
 U^i(X^i;Z)&:=\sup_{k^i\in \mathcal{C}_i}E_{P^i}[u^i(X^i+k^i+Z)] \notag \\
& =\min_{\lambda >0} \min_{\substack{  \probq\ll\probp^i,\\ Q\in M_{\sigma}(S^i)}}
\Bigg (\lambda (E_Q[X^i]+E_Q[Z])+E\bigg[\Phi^i\Big(\lambda \rn{\probq}{\probp^i}\Big)\bigg] \Bigg )  \label{minV}
\end{align}

The Assumption $M_{e,\mathrm{loc}} (S^i)\cap \mathcal P_{\Phi^i}(P^i) \neq \emptyset$ also imply (see Theorem 2 \cite{F00} and Remark 2 in \cite{BF05}) that the optimizer in the above minimization problem is an \textit{equivalent} probability measure, namely 
 \begin{equation*}
        U^i(X^i;Z)=\min_{\lambda >0} \min_{ Q\in M_{e,\sigma}(S^i) }
\Bigg (\lambda(E_Q[X^i]+E_Q[Z])+E\bigg[\Phi^i\Big(\lambda \rn{\probq}{\probp^i}\Big)\bigg] \Bigg ) ,\quad  \forall\,Z \in L_i\text{ s.t. } U^i(X^i;Z)<u^i(+\infty).      
\end{equation*}

Since $U^i(X^i;0)<u^i(+\infty)$ we see that $U^i(X^i;\cdot)$ is norm continuous on $L_i$ by Proposition \ref{propU*}. In particular, $U^i(X^i;c 1_A)<U^i(+\infty)$ for $c$ small enough, i.e. this holds whenever $0<c<C$ for a suitable $C>0$. For all such values of $c$,
from \eqref{minV}
we obtain
\begin{align*}
U^i(X^i;c 1_A)&=\min_{\lambda >0} \min_{\substack{  \probq\ll\probp^i,\\ Q\in M_{e,\sigma}(S^i)}}
\Bigg (\lambda (E_Q[X^i] +c Q(A) )+E_{P^i}\bigg[\Phi^i\Big(\lambda \rn{\probq}{\probp^i}\Big)\bigg] \Bigg )    \\
&=\widehat\lambda (E_{\widehat{Q}}[X^i] +c  \widehat{Q}(A)) +E_{P^i}\bigg[\Phi^i\Big(\widehat\lambda \rn{\widehat\probq}{\probp^i}\Big)\bigg] \Bigg )
 >\widehat\lambda E_{\widehat{Q}}[X^i]+E_{P^i}\bigg[\Phi^i\Big(\widehat\lambda \rn{\widehat\probq}{\probp^i}\Big)\bigg] \Bigg ) \\
&\geq \min_{\lambda >0} \min_{\substack{  \probq\ll\probp^i,\\ M_{\sigma}(S^i)}}
\Bigg (\lambda E_Q[X^i]  +E_{P^i}\bigg[\Phi^i\Big(\lambda \rn{\probq}{\probp^i}\Big)\bigg] \Bigg ) =U^i(X^i;0),
\end{align*}
as the optimizer $\widehat{Q}$ satisfies $\widehat{Q}\sim P^i\sim \Pp$ and thus $\widehat{Q}(A)>0$ 
\end{proof}

\subsection{Mackey topology and utility maximization}

\begin{proposition}
\label{mackeycont}
    Suppose that Assumption \ref{assOK} holds,  let $X \in L_i=L^\infty(\Omega,\mathcal{F}^i,P^i)$ and set $ U^i(X):=\sup_{k\in\mathcal{C}_i}E_{P^i}[u^i(X+k)]$ ( as in \eqref{eqUW}). Suppose that $U^i(\overline{X})<+\infty$ for some $\overline{X}\in L_i$.
    Then \(U^i\) is Lebesgue continuous on $L^\infty(\Omega,\mathcal{F}^i,P^i)$, namely for every bounded sequence \((Z_n)_{n}\subset L^\infty(\Omega,\mathcal{F}^i,P^i)\) with \(Z_n\to Z\in L^\infty(\Omega,\mathcal{F}^i,P^i)\) pointwise a.s. we have $
U^i(Z_n)\to_n U^i(Z).$
  Furthermore, $U^i$ is continuous in the Mackey topology $\tau(L^\infty(\Omega,\mathcal{F}^i,P^i),L^1(\Omega,\mathcal{F}^i,P^i))$.
\end{proposition}
\begin{proof}
$U^i$ can be easily defined on $ M^{\hat u^i}(\Omega,\mathcal{F}^i,P^i)$.
    For every $X\in M^{\hat u^i}(\Omega,\mathcal{F}^i,P^i)$ we have $U^i(X)>-\infty$, and since $L^\infty(\Omega,\mathcal{F}^i,P^i)\subseteq M^{\hat u^i}(\Omega,\mathcal{F}^i,P^i)$ we have that $U^i$ is finite on a point of $M^{\hat u^i}(\Omega,\mathcal{F}^i,P^i)$. It is easy to verify that $U^i$ is concave on $M^{\hat u^i}(\Omega,\mathcal{F}^i,P^i)$ and monotone w.r.t. the a.s. ordering. Hence, it is finite on the whole Orlic Heart and norm continuous on it by Extended Namioka-Klee Theorem \cite{bfnam} Theorem 1. Take a sequence \((Z_n)_{n}\subset L^\infty(\Omega,\mathcal{F}^i,P^i)\) with \(Z_n\to Z\in L^\infty(\Omega,\mathcal{F}^i,P^i)\) pointwise a.s. and $\sup_n\|Z_n\|_{\infty}<K$. By order continuity of the Luxemburg norm on the Orlicz Heart (\cite{Edgar} Theorem 2.1.14) we have $Z_n\rightarrow_n Z$ in norm.
Setting $A_n=\sup_{m\geq n}\abs{Z_n-Z}$ we have that $A_n$ is still bounded in  $L^\infty(\Omega,\mathcal{F}^i,P^i)$, since so is the original sequence $(Z_n)_n$,  $A_n\downarrow_n 0$ a.s. clearly, and $A_n\geq \abs{Z_n-Z}$. By order continuity of the norm we have
    $$0=\lim_n\norm{A_n}_{\widehat{u}^i}\geq \limsup_n\norm{Z_n-Z}_{\widehat{u}^i}.$$
    By norm continuity of $U^i$ then \(
U^i(Z_n)\to_n U^i(Z).\)
We conclude that the Lebesgue property holds.
This in turns implies that $f=-U^i:L^\infty(\Omega,\mathcal{F}^i,P^i) \rightarrow \R$ is continuous for the topology of convergence in probability on bounded sets of $L^\infty(\Omega,\mathcal{F}^i,P^i)$. Thus, \cite{DelbaenOwari19} Proposition 1.2 provides Mackey continuity of $f$, and obviously of $U^i$ as a consequence.
\end{proof}
\subsection{Local martingales in discrete time}

Let  $M_e(S)$ (resp. $M_{e,\mathrm{loc }}(S)$, $M_{e,\sigma}(S)$) be the set of equivalent martingale (resp. local martingale, sigma martingale) measures for $S$.

\begin{proposition}\label{propMsigma}
Consider a finite discrete time setting $\mathcal T:=\{0,1,\dots,T\} $,  and a filtered probability space $(\Omega, \mathcal F, \mathbb F=(\mathcal F_t)_{t \in \mathcal T},P)$ with the initial
$\sigma$-algebra $\mathcal F_0$ trivial and $\mathcal F_T=\mathcal F$.
Let $(S_t)_{t\in \mathcal T}$ be a  $d$-dimensional $\mathbb F$-adapted process. Then
\begin{enumerate}
    \item $S$ is a local martingale if and only if it is a sigma martingale.
    \item \label{changetomtg} If $S$ is a local martingale then there exists a probability $Q$ equivalent to $P$ such that $S$ is a $(\mathbb F,Q)$ martingale.
    \item \label{lociffmart} $M_e(S) \neq \emptyset$ if and only if $M_{e, \mathrm{loc}}(S) = M_{e, \sigma}(S) \neq \emptyset$.
    \item If $S$ is an integrable local martingale then it is a martingale.
    \item \label{bddblowismart} If $S$ is a local martingale with $S_T \geq 0$ then it is a martingale. 
     \item \label{integrisloc}  Suppose $S$ is a local martingale. Let $H=(H_t)_{t \in \mathcal T}$ be a predictable process and let $(H \cdot S)=((H \cdot S)_t)_{t \in \mathcal T}$ be the stochastic integral:
\[
(H \cdot S)_t := \sum_{s= 1}^t H_s(S_s-S_{s-1}), \, \text{ for }t \geq 1; \quad 
(H \cdot S)_0=0.
\]
Then $(H \cdot S)$ is a local martingale.\\
If in addition we assume that $(H \cdot S)_T\ge -C$
for some constant $C>0$, then  $(H \cdot S)$ is a martingale, and in particular  $E_P[(H \cdot S)_t]=0$ for all $t\in\mathcal{T}$.
    
\end{enumerate}
\end{proposition}

\begin{proof}

\begin{enumerate}
\item Follows from Theorem 1 in\cite{JaShi98}. Notice that Theorem 1 is stronger, as it shows the equivalence between local martingale and martingale transform, where the predictable process $\gamma$ in equation (8) in the martingale transform need not to be strictly positive. See also the three lines  after the proof of the cited theorem, which states that the same results holds in the finite discrete time setting.
\item Is Theorem 1 in \cite{Ka08}.

\item Follows from Items 1 and 2.

\item Is Proposition 1.3.14 in \cite{KaSH91}

\item See page 295 in \cite{Ka08}.
\item The first statement follows from Theorem 1 in \cite{JaShi98}. For the last statement, we just showed that the stochastic integral is itself a local martingale.
By Proposition \ref{propMsigma} Item \ref{bddblowismart} the process $X_t=(H \cdot S)_t+C$ is a local martingale satisfying $X_0 \in \R$ and $X_T \geq 0$, thus it is a martingale, and so is $(H \cdot S)$.

\end{enumerate}
\end{proof}

\section{Appendix: Orlicz spaces}
\label{approlicz}
Orlicz spaces provide a natural extension of the classical $L^p(\Omega,\mathcal{F},\Pp)$ spaces and
offer a flexible alternative to the restrictive $L^\infty(\Omega,\mathcal{F},\Pp)$ setting.
They allow for unbounded payoffs while retaining good compactness and duality
properties. In particular, when one works on the Orlicz Heart, the dual space
does not contain singular elements, which greatly simplifies dual
representations in utility maximization and risk-measure theory.
We refer to \cite{Edgar} or \cite{RaoRen} for further theoretical background.

Let $u:\mathbb{R}\to\mathbb{R}$ be a concave and nondecreasing utility function,
satisfying the standard left-tail growth condition
\begin{equation}\label{eq:orlicz_asymptotic}
\lim_{x\to -\infty}\frac{u(x)}{x}=+\infty.
\end{equation}
We associate to $u$ the function
\begin{equation}\label{eq:uhat_def}
\widehat{u}(x):=-u(-|x|)+u(0), \qquad x\in\mathbb{R}.
\end{equation}
The function $\widehat{u}$ is a strict Young function: it is finite-valued,
even, convex, satisfies $\widehat{u}(0)=0$, and grows superlinearly at infinity,
that is,
\[
\lim_{x\to+\infty}\frac{\widehat{u}(x)}{x}=+\infty.
\]
The Orlicz space generated by $\widehat{u}$ is defined as
\begin{equation}\label{eq:orlicz_space}
L^{\widehat{u}}(\Omega,\mathcal{F},\Pp)
:=\Bigl\{X\in L^0(\Omega,\mathcal{F},\Pp)\;:\;
E_{\Pp}[\widehat{u}(\alpha \abs{X})]<\infty
\text{ for some }\alpha>0\Bigr\}.
\end{equation}
The corresponding Orlicz Heart (also known as the Morse-- space) is
\begin{equation}\label{eq:orlicz_heart}
M^{\widehat{u}}(\Omega,\mathcal{F},\Pp)
:=\Bigl\{X\in L^0(\Omega,\mathcal{F},\Pp)\;:\;
E_{\Pp}[\widehat{u}(\alpha \abs{X})]<\infty
\text{ for every }\alpha>0\Bigr\}.
\end{equation}
Endowed with the Luxemburg norm
\[
\|X\|_{\widehat{u}}
:=\inf\Bigl\{k>0:\ E_{\Pp}\bigl[\widehat{u}(|X|/k)\bigr]\le 1\Bigr\},
\]
both $L^{\widehat{u}}(\Omega,\mathcal{F},\Pp)$ and $M^{\widehat{u}}(\Omega,\mathcal{F},\Pp)$ are Banach spaces. Moreover, the
standard inclusions hold:
\[
L^\infty(\Omega,\mathcal{F},\Pp)
\subseteq M^{\widehat{u}}(\Omega,\mathcal{F},\Pp)
\subseteq L^{\widehat{u}}(\Omega,\mathcal{F},\Pp)
\subseteq L^1(\Omega,\mathcal{F},\Pp).
\]
Let $\widehat{u}^\ast$ denote the convex conjugate of $\widehat{u}$,
\[
\widehat{u}^\ast(y):=\sup_{x\in\mathbb{R}}\{xy-\widehat{u}(x)\},
\qquad y\in\mathbb{R}.
\]
Fenchel's inequality yields
\[
xy\le \widehat{u}(x)+\widehat{u}^\ast(y), \qquad x,y\in\mathbb{R}.
\]
As a consequence, if $Q\ll\Pp$ and $\rn{Q}{\Pp}\in L^{\widehat{u}^\ast}(\Omega,\mathcal{F},\Pp)$,
then every $X\in L^{\widehat{u}}(\Omega,\mathcal{F},\Pp)$ is $Q$-integrable, that is,
$L^{\widehat{u}}(\Omega,\mathcal{F},\Pp)\subset L^1(\Omega,\mathcal{F},Q)$. When working on the Orlicz Heart
$M^{\widehat{u}}(\Omega,\mathcal{F},\Pp)$, this duality is particularly well behaved and underlies
many classical results in utility maximization.
By construction of $\widehat{u}$, the condition
\[
E_{\Pp}[\widehat{u}(X)]<\infty
\]
implies the usual admissibility requirement
\[
E_{\Pp}[u(X)]>-\infty.
\]

{
\bibliographystyle{abbrv}  
\bibliography{BibAll}
}

\end{document}